%% file: OrientLocalDistr_arXiv.tex
\documentclass[11pt]{article}
\RequirePackage{algorithm,algorithmic,graphicx,amssymb,epsf,epic,url}
\RequirePackage{fullpage}

\newtheorem{question}{Question}

\newcommand{\local}{{\cal LOCAL}}
\newcommand{\congest}{\mathcal{CONGEST}}
\newcommand{\Label}{\mbox{\tt Label}}
\newcommand{\id}{\mbox{\tt ID}}
\newcommand{\NULL}{\textsc{null}}
\input{latexmac}

\denseformat

\newcommand {\ignore} [1] {}

\begin{document}

\title{Dynamic Representations of Sparse Distributed Networks: \\A Locality-Sensitive Approach}
\author{
Haim Kaplan \thanks{School of Computer Science, Tel Aviv University, Tel Aviv 69978, Israel.
E-mail: {\tt haimk@post.tau.ac.il}.}
\and Shay Solomon \thanks{IBM Research.
E-mail: {\tt solo.shay@gmail.com}.}}

\date{\empty}

\begin{titlepage}
\def\thepage{}
\maketitle

\begin{abstract}
In 1999, Brodal and Fagerberg (BF) gave an algorithm for maintaining a low outdegree orientation of a  dynamic uniformly sparse graph.
Specifically, for a dynamic graph on $n$-vertices, with arboricity bounded by $\alpha$ at all times,
the BF algorithm supports edge updates in $O(\log n)$ amortized update time, while keeping the maximum outdegree in the graph bounded by $O(\alpha)$.
Such an orientation
provides a basic data structure for uniformly sparse graphs, which found applications to several dynamic graph algorithms,
including adjacency queries and labeling schemes, maximal and approximate matching, approximate vertex cover, forest decomposition, and distance oracles.

A significant weakness of the BF algorithm is the possible \emph{temporary} blowup of the maximum outdegree, following edge insertions.
Although BF eventually reduces all outdegrees to $O(\alpha)$, some vertices may reach an outdegree of $\Omega(n)$ during the process,
hence local memory usage at the vertices, which is an important quality measure in distributed systems, cannot be bounded.
We show how to modify the BF algorithm to guarantee that the outdegrees of all vertices are bounded by $O(\alpha)$ at all times, without hurting any
of its other properties, and present an efficient distributed implementation of the modified algorithm.
This provides the \emph{first} representation of distributed networks in which the local memory usage at all vertices is bounded by the arboricity
(which is essentially the average degree of the densest subgraph) rather than the maximum degree.

For settings where there is no strict limitation on the local memory,
one may take the temporary outdegree blowup to the extreme and allow a permanent outdegree blowup.
This allows us to address the second significant weakness of the BF algorithm -- its inherently \emph{global} nature:
An insertion of an edge $(u,v)$ may trigger changes in the orientations of edges that are arbitrarily far away from $u$ and $v$.
Such a non-local scheme may be prohibitively expensive in various practical applications.
We suggest an alternative \emph{local} scheme, which does not guarantee any outdegree bound on the vertices,
yet is just as efficient as the BF scheme for some of the aforementioned applications.
For example, we obtain a local dynamic algorithm for maintaining a maximal matching with sub-logarithmic update time in uniformly sparse networks,
providing an exponential improvement over the state-of-the-art in this context.
We also present a distributed implementation of this scheme and some of its applications.
\end{abstract}
\end{titlepage}

\pagenumbering {arabic}

\section{Introduction} \label{intro}

\subsection{Quality measures in distributed computing} \label{sec11}
The $\local$ and the $\congest$ models are perhaps the two most fundamental communication models in distributed computing (cf.~\cite{Peleg00}),
the former is the standard model capturing the essence of spatial locality, and the latter also takes into account congestion limitations.
In these models it is assumed that initially all the processors wake up simultaneously,
and that computation proceeds in fault-free synchronous rounds during which every processor exchanges messages with its direct neighbors in the network.
In the $\local$ model these messages are of unbounded size, whereas in the $\congest$ model each message contains $O(\log n)$ bits.
An efficient distributed algorithm allows the nodes to communicate with their direct neighbors for a small number of rounds,
after which they need to produce their outputs, which are required to form a valid global solution.
A task is called \emph{local} if the number of rounds needed for solving it is constant.
The locality of many distributed tasks have been studied in the past two decades, with the emerging conclusion that truly local tasks are rather scarce.

Another important locality measure is the local \emph{memory usage} at a processor.
The standard premise is that each processor may communicate with \emph{all} its neighbors by sending and receiving messages.
To this end, the local memory usage at a processor should be proportional to (and at least linear in) its degree.
Reducing the local memory at processors to be independent of their degree could be of fundamental importance for many real-life applications.
In fact, the processors in a distributed network are in many cases \emph{identical},
thus the local memory at low degree processors is not proportional to their degree but rather to the \emph{maximum degree} in the network.
Moreover, in sparse networks (such as planar networks), the maximum degree may be $n-1$ while the average degree is constant,
so the global memory (over all processors) will be blown up by a factor of $n$ if all the processors are identical.
(In dynamic networks, on which we focus here, this factor $n$ blow-up may occur even if the processors are not identical.)
Low-degree spanners have been used to reduce local memory usage at processors,
which was proved useful for a plethora of applications, such as
efficient broadcast protocols, data gathering and dissemination tasks in overlay networks, compact routing schemes, network synchronization, computing global functions
\cite{Awerbuch85,PU89,ABP90,ABP91,Peleg00}.
However, for the vast majority of distributed tasks, such as maximum independent set  and coloring, the global solution must consider
\emph{all edges} of the network and not just the spanner edges.

The total number of messages needed for solving a distributed task is another fundamental quality measure in distributed computing,
which we will also consider in the sequel.
\subsection{The dynamic distributed setting} \label{sec12}
The dynamic distributed model is defined as follows. Starting with the empty graph $G_0=(V,E_0=\{\emptyset\})$, in every
round $t>0$, the adversary chooses a vertex or an edge to be either inserted to or deleted from $G_{t-1}$, resulting in $G_t$.
(As a result of a vertex deletion, all its incident edges are deleted. A vertex is inserted without incident edges.)
Upon the insertion or deletion of a vertex $v$ or an edge $e=(u,v)$, an update procedure is invoked, which should restore the validity of the solution being maintained.
For example, if we maintain a maximal matching, then following the deletion of a matched edge the matching is no longer maximal, and the update procedure should restore maximality. We shall consider the most natural model in this setting, hereafter the \emph{local wakeup model} (cf.\ \cite{PS16,PPS16,CHK16,AOSS18}),
where only the affected vertices wake up (following an update to a vertex $v$, only $v$ wakes up; following an edge update $(u,v)$, both $u$ and $v$ wake up).
The update procedure proceeds in fault-free synchronous rounds during which every processor exchanges messages with its neighbors, just as in the static setting, until finishing its execution.

In the distributed dynamic setting, the \emph{amortized update time} and \emph{amortized message complexity} bound the \emph{average} number of communication rounds
and messages sent, respectively, needed to update the solution per update operation, over a \emph{worst-case} sequence of updates.
The \emph{worst-case update time} and \emph{worst-case message complexity} is the maximum number of communication rounds and messages sent,
again over a worst-case sequence of updates.

We assume that the topological changes occur serially and are sufficiently spaced so that the protocol has enough time to complete its operation before the occurrence of the next change.
Since all our algorithms can be strengthened to achieve a worst-case update time of $O(\log n)$ (and in some cases even $O(1)$),
this assumption should be acceptable in many practical scenarios.
Moreover, the same assumption has been made also in previous works; see, e.g., \cite{PS16,PPS16,CHK16,AOSS18},  and the references therein.
We   remark that our focus is on optimizing \emph{amortized} rather than worst-case bounds, which may provide another justification for making this assumption.
\subsection{Representations of sparse networks via dynamic edge orientations} \label{sec13}
\subsubsection{Centralized networks} \label{sec131}
A graph $G=(V,E)$ has \emph{arboricity} $\alpha$ if $\alpha=\max_{U\subseteq V}\left\lceil\frac{|E(U)|}{|U|-1}\right\rceil$,
where $E(U)=\left\{(u,v)\in E\mid u,v\in U\right\}$. Thus the arboricity is close to the maximum
\emph{density} ${|E(U)|}/{|U|}$ over all induced subgraphs
of $G$.
While a graph of bounded arboricity is \emph{uniformly sparse}, a graph of bounded density (i.e., a sparse graph) may contain a dense subgraph
(e.g., on $\sqrt{m}$ of the vertices), and therefore may have large arboricity.
The family of bounded arboricity graphs
contains planar and bounded genus graphs, bounded tree-width graphs, and in general all graphs excluding fixed minors.

One of the most fundamental questions in data structures is to devise efficient representations of graphs supporting \emph{adjacency queries}:
Given two vertices $u$ and $v$, is there an edge between them in the $n$-vertex graph $G = (V,E)$?
Using an adjacency matrix (of size $\Theta(n^2)$) one can support such queries in $O(1)$ time. In sparse graphs, however, a quadratic-space data structure
seems very wasteful.
If one uses adjacency lists instead, the space is reduced to $O(|E|)$, but then adjacency queries may require $\Theta(n)$ time.
By maintaining these adjacency lists sorted, the worst-case query time can be reduced to $O(\log n)$, but no further than that, even in sparse graphs.
Another approach is to use hashing, which guarantees linear space and constant query time, but alas it requires randomization,
otherwise the construction time is super-linear.
While some of these data structures have linear (in the graph size) space usage, none of them can bound the \emph{local} space usage (per vertex).

In a pioneering paper from 1999, Brodal and Fagerberg (BF) \cite{BF99} devised a data structure for adjacency queries in uniformly sparse graphs that is based on \emph{edge orientations}.
Specifically, an \emph{arboricity $\alpha$ preserving sequence}
is a sequence of edge insertions and deletions starting from an empty graph, in which the arboricity of the dynamic graph is bounded by $\alpha$ at all times.
For any arboricity $\alpha$ preserving sequence, the BF algorithm has an amortized time update time of $O(\log n)$,
while keeping the maximum outdegree in the graph bounded by $\Delta = O(\alpha)$.
(The BF algorithm can, in fact, handle vertex updates within the same asymptotic bounds, where $n$ stands for the current number of vertices.)
Such an edge orientation, which is called a \emph{$\Delta$-orientation},
allows to support adjacency queries in $O(\alpha)$ worst-case time,
thus providing a significant improvement over the known data structures in graphs of sufficiently low arboricity.

BF also showed that the amortized time of their algorithm is asymptotically optimal.
Specifically, let $\alpha,\delta$ and $\Delta$ be three arbitrary integers satisfying $\alpha \ge 1, \delta = \Omega(\alpha), \Delta = \Omega(\delta)$,
and suppose one can maintain a $\delta$-orientation for some sequence of $t$ edge updates while doing $f$ edge flips, starting with the empty graph.
(We omit the constants hidden in the $\Omega$ notation above and the $O$ notation to follow.)
Then the BF algorithm on this update sequence with an outdegree parameter $\Delta$ maintains a $\Delta$-orientation with a total runtime (and thus number of edge flips) of $O(t+f)$.

Recently, there has been a growing interest in the edge orientation problem, due to its applications to additional dynamic graph problems.
See App.\ \ref{discuss} for additional results on this problem and some of its applications.
\subsubsection{Distributed networks} \label{sec132}
There is a close connection between low outdegree orientations and the forest decomposition problem,
where one aims to decompose the edges of a graph $G$ into a small number of (rooted) forests.
Obviously, a decomposition of a graph into $\ell$ forests immediately yields an $\ell$-orientation.
The other direction is also true \cite{PPS16}: An $\ell$-orientation yields a decomposition into at most $2\ell$ forests.
Also, a dynamic maintenance of the former can be translated into a dynamic maintenance of the latter with a constant overhead in the update time,
in both centralized and distributed settings \cite{PPS16}.

\cite{BE10} studied the forest-decomposition problem in the distributed static setting.
They showed that for a network $G$ with arboricity $a(G)$ and any $q>2$, there exists a distributed algorithm that computes a decomposition of $G$ into at most $((2+q) \cdot a(G))$ forests (and hence also a $((2+q) \cdot a(G))$-orientation) in
$O(\frac{\log n}{\log q})$ rounds. (This result was refined recently by \cite{GS17}.)
\cite{BE10} also showed that given such a forest decomposition (or an edge orientation), one can compute an $O(q \cdot a(G)^2)$-vertex coloring for $G$ in $O(\log^* n)$ more rounds. Using this coloring an MIS can be computed in $O(q \cdot a(G)^2)$ rounds.
More generally, low outdegree orientations lead to sublinear-time algorithms for vertex and edge coloring, MIS, and maximal matching in distributed networks of bounded arboricity.
(See Chapters 4 and 11.3 in  \cite{BE13} for more details.)

For the dynamic distributed model, \cite{PPS16} devised a distributed algorithm for maintaining
$O(\alpha + \log^* n)$-orientation in $O(\log^* n)$ amortized update time. They then used this orientation to maintain
within the same time a decomposition into $O(\alpha + \log^* n)$ forests and also an adjacency labeling scheme with label size $O(\alpha + \log^* n)$.
They used the same approach to get distributed algorithms for maintaining $O(\alpha \cdot \log^* n)$-coloring and other related structures with the same $O(\log^* n)$
update time. Although the distributed algorithm of \cite{PPS16} has a low amortized update time,
it incurs a polynomial (in the network size) bound on three important parameters: (1) the amortized message complexity, (2) the local memory usage at processors, and (3) the messages size. In particular, the algorithm of \cite{PPS16} cannot be implemented in the $\congest$ model.

While the distributed algorithms of \cite{BE10} can be implemented in the $\congest$ model,
they are static, and as such, their message complexity must be at least linear in the size of the network.
Moreover, unless there is some underlying representation of the network,
for an algorithm to solve any nontrivial task from scratch, any processor must communicate with each of its neighbors at least once.
Hence the local memory usage at processors, which should be at least linear in the maximum degree for some processors,
may be larger than the arboricity bound $\alpha$ by a factor of $n / \alpha$.
\vspace{6pt}
\\
{\bf A fundamental question.~} Can one use $O(\alpha)$-orientations to obtain a representation of a dynamic distributed network with a local memory usage of $O(\alpha)$?
We first argue that a distributed implementation of the BF algorithm cannot achieve this.
Indeed, a significant weakness of the BF algorithm is the possible \emph{temporary} blowup of the maximum outdegree, following edge insertions.
More specifically, following an insertion of edge $(u,v)$ that is oriented from $u$ to $v$, the outdegree of $u$ may exceed the  threshold $\Delta$.
To restore a valid $\Delta$-orientation, the BF algorithm \emph{resets} $u$, thereby flipping all its outgoing edges.
As a result, the former  \emph{out-neighbors (outgoing neighbors)} of $u$ increase their outdegree.
All such neighbors whose outdegree now exceeds $\Delta$ are then handled in the same way, one after the other,
and this process is repeated until all vertex outdegrees are $\le \Delta$.
BF used an elegant potential function argument to show that this process not only terminates,
but also leads to an asymptotically optimal algorithm (as mentioned before).
Although BF eventually reduces all outdegrees to $\le \Delta$, some of these outdegrees may blow up throughout the reset cascade all the way to $\Omega(n)$.

To implement the BF algorithm with local memory usage of $O(\alpha)$,
the orientation should remain a $\Delta$ (or close to $\Delta$)-orientation throughout the reset cascade.
We show that this is not the case unless the graph is of arboricity $1$.
Specifically, we show that
for dynamic forests ($\alpha = 1$), the BF algorithm never increases the outdegree of a vertex beyond $\Delta+1$,
but there exist graphs of arboricity $2$ for which the BF algorithm blows up the outdegree of some vertices to $\Omega(n)$!
Hence, a distributed implementation of the BF algorithm requires a huge local memory usage.
The algorithms of \cite{KKPS14,HTZ14}, with a worst-case update time, never increase the outdegree of a vertex beyond the specified threshold.
However, the tradeoffs between the outdegree and update time provided by these algorithms are significantly inferior to the BF tradeoff.
In particular, for graph of constant arboricity, the outdegree should remain constant at all times,
and the algorithms of \cite{KKPS14,HTZ14} cannot provide outdegree better than $\Omega(\log n/ \log\log n)$. (See App.\ \ref{discuss} for more details.)

We remark that the reset cascade of the BF algorithm is inherently sequential,
and it is unclear if it can be distributed efficiently
even regardless of local memory constraints.
A similar issue arises with the worst-case update time algorithms of \cite{KKPS14,HTZ14}.  
\vspace{4pt}
\vspace{-4pt}
\begin{question} \label{q3}
Is there an algorithm with the same optimal (up to constants) tradeoff of BF between the outdegree $\Delta$ and the amortized cost,
which guarantees that the outdegree of all vertices is always $O(\Delta)$?
Furthermore, can this algorithm be distributed efficiently with a local memory usage of $O(\Delta)$?
\end{question}
\vspace{6pt}
{\bf Our contribution.~}
Our first attempt towards answering Question \ref{q3} is by making a natural modification to the BF algorithm:
Instead of resetting vertices of outdegree larger than $\Delta$ at an arbitrary order, we always choose
to reset next, the vertex of largest outdegree among all vertices of outdegree larger than $\Delta$.
We show that with this modification the algorithm of BF keeps the outdegrees
 $O(\Delta \log(n/\Delta))$ at all times.
We also complement this upper bound with a matching lower bound,
showing that
the BF algorithm together with this modification can indeed generate vertices of outdegree $\Omega(\Delta \log(n/\Delta))$
during the reset cascade, and this can happen even in  graphs of arboricity 2.
This modification  does not resolve Question \ref{q3}, as the outdegree may blow-up by a logarithmic factor during the cascade,
and more importantly, 
it seems unlikely that the algorithm with this modification can be distributed efficiently.

To resolve Question \ref{q3} we first give a new \emph{centralized} algorithm, which is inherently different than the BF algorithm, and keeps the outdegree bounded by
 $O(\Delta)$ at all times.
In contrast to the BF algorithm, our algorithm does not apply a cascade of reset operations on vertices whose outdegree exceeds $\Delta$ following an insertion.
Note that any reset operation on some vertex ``helps'' that particular vertex but ``hurts'' its out-neighbors.
Instead, our algorithm first collects a set of vertices of relative high outdegree that would ``benefit'' from being reset.
Then it works on the graph $G^*$ induced by the outgoing edges of these vertices in a somewhat opposite manner to the BF algorithm.
More specifically, it applies a cascade of ``anti-reset'' operations on vertices of outdegree significantly smaller than $\Delta$,
where an anti-reset on a vertex flips all its incoming edges to be outgoing of it.
In other words, vertices in our algorithm are being helpful to their neighbors rather than hurtful as before.
The cascade of ``anti-reset'' operations leads to a low outdegree orientation within the subgraph $G^*$,
but it also makes sure that the outdegree of all vertices would never exceed $\Delta + 1$ in the entire graph throughout the process.  
We show that our algorithm has the same (up to a constant factor) tradeoff of BF between the outdegree and amortized cost.
This is nontrivial, since the potential function argument of BF relies heavily on the gain of \emph{any} reset operation to the potential value.
Roughly speaking, that argument compares the current orientation to an optimal orientation, where all edges but $\approx \alpha$ must be incoming to any vertex, and so the potential must be reduced after resetting a vertex of outdegree much larger than $\alpha$. This argument, alas, does not carry over to anti-resets.
The argument that we provide is based on a global consideration (of the total potential gain of all anti-resets) rather than on a local consideration (of each reset).
We also demonstrate that this approach of replacing resets with anti-resets facilitates efficient distributed implementation,
as we can perform all the anti-resets in parallel, without worrying about the neighbors' outdegrees.

In this way we resolve Question \ref{q3} in the affirmative, providing a distributed algorithm for maintaining $\Delta$-orientation with the optimal (up to constants, w.r.t.\ $\Delta$)
amortized cost,
with a local memory usage of $O(\Delta)$, for any $\Delta = \Omega(\alpha)$.
Moreover, the amortized cost bounds not just the amortized update time of our algorithm but also its amortized message complexity.
Our algorithm uses short messages, and can thus be implemented in the $\local$ model.
As immediate consequences, we can maintain forest decomposition and adjacency labeling schemes with the same bounds as above,
thereby significantly improving \cite{PPS16}. (Recall that the algorithm of \cite{PPS16} incurs polynomial bounds on the amortized message complexity, local memory usage at processors, and messages size.)

A low outdegree orientation does not provide information on the incoming neighbors of a vertex.
Hence, although it finds applications as discussed above, it cannot be viewed as a complete representation of the network.
To obtain a complete representation of the network, we distribute the information on the incoming neighbors of any vertex $v$ within the local memory of these neighbors.
In this way we can guarantee that the local memory usage remains $O(\Delta)$, yet each vertex can scan its incoming neighbors upon need.
On the negative side, this scan of incoming neighbors will be carried out sequentially rather than in parallel.
Nevertheless, in some applications, we only need to scan a few incoming neighbors.
As a first application of our network representation, we obtain a distributed algorithm for maintaining a maximal matching with $O(\log n)$ amortized update time and message complexities,
with $O(\alpha)$ local memory usage. (A maximal matching can be maintained via a trivial distributed algorithm with $O(1)$ worst-case update time, even in general networks,
but its amortized message complexity and local memory usage will be $\Omega(n)$, even in forests.)
To enhance the applicability of our network representation, we demonstrate that the \emph{bounded degree sparsifiers} of \cite{Sol18} can be maintained dynamically in a distributed network using low local memory usage.
Using these sparsifiers, we obtain efficient distributed algorithms for maintaining approximate matching and vertex cover
with low amortized update time and message complexities and with low local memory usage (see Section \ref{sec2} for details).


This result provides the first efficient representation of uniformly sparse distributed networks with low local memory usage.
Besides the aforementioned applications, such a representation may be used more broadly in applications currently suitable only for low degree networks,
where local memory is very limited.
\subsection{The algorithm of BF is global} \label{sec14}
When dealing with networks of huge scale, it is often important to devise algorithms that are intrinsically {\em local}.
\emph{Local algorithms} have been extensively studied, from various perspectives.
(See e.g.~\cite{LPP08,ARVX12,Suomela13,RTVX11,MV13,EMR14,EMR15} and the references therein.)
A local algorithm in a dynamic network performs an operation at a vertex $v$ while affecting only $v$ and its immediate neighbors (or more generally vertices in a small ball around $v$).
Local algorithms are motivated by environments, both centralized and distributed, in which it is
undesirable, and sometimes even impossible, for a change at a particular vertex of the network to affect
remote locations unrelated to the change. In the context of I/O efficiency, local algorithms may have better cache performance.


The second drawback of the  BF algorithm that we address is the fact that it is not local.
A single edge insertion $e = (u,v)$ that increases the outdegree of a vertex  beyond
 $\Delta$ may trigger edge flips that are at distance $\Theta(\log_{\Delta} n)$ from $u$ and $v$,
as shown in Figure \ref{f:farflips} for $\Delta = 2$.
In fact, for the example of Figure \ref{f:farflips},
\emph{any} algorithm that maintains a $\Delta$-orientation
must flip edges that are at distance $\Theta(\log_{\Delta} n)$ from $u$ and $v$.
(There are degenerate examples showing that the BF algorithm sometimes flips edges at distance $\Theta(n)$ from $u$ and $v$.)
Consequently, to achieve locality, we must relax the outdegree condition inherent to the edge orientation problem.
\begin{figure}[h!]
\begin{center}
\includegraphics[scale=0.7]{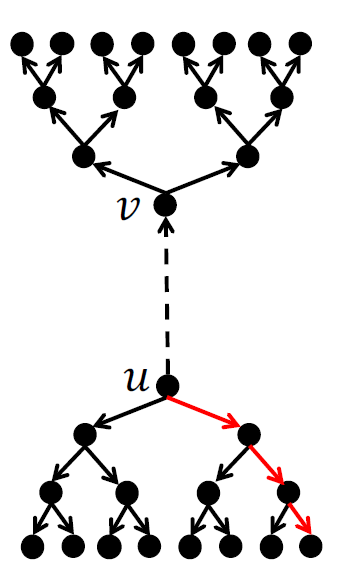}
\end{center}
\caption{An illustration for $2$-orientation. Upon the insertion of edge $(u,v)$, at least $\Omega(\log_2 n)$ edges must be flipped to restore a 2-orientation,
some of which must be at distance $\Omega(\log_2 n)$ from $u$ and $v$.
For example, flipping the $\Theta(\log n)$ edges along the red path should restore a $2$-orientation.\label{f:farflips}}
\end{figure}


\vspace{6pt}

\noindent
{\bf Our contribution.~}
We propose an alternative \emph{local} scheme that performs a sequence of edge insertions, deletions, and adjacency queries in total time that is
asymptotically no worse than that of BF. The scheme is natural and works as follows.
Upon a query and/or an update at a vertex $v$  we reset $v$. That is we make $v$'s outgoing edges incoming. (We suggest two versions, one aggressive that always
flips  $v$'s outgoing edges,  and another that flips these edges only if the outdegree of $v$ is larger than $\Delta$.)
More specifically, whenever the application of interest has to traverse $v$'s outgoing neighbors it also flips them (thereby intuitively paying for the traversal). Thus, we
get locality at the cost of giving away the worst case upper bound on the outdegrees of the vertices.
We call this scheme the \emph{flipping game}.
We use the flipping game to get  local algorithms for adjacency queries and dynamic maximal matching.
These two application can, in fact, be casted as special cases of a generic paradigm, described in detail in Section \ref{appgen}.


The only known local algorithm for maintaining maximal matching
has update time of $O(\sqrt{m})$ where $m$ is the number of edges in the graph~\cite{NS13}, and
this guarantee does not improve for graphs with bounded arboricity.
(Even in dynamic forests, the fastest known local algorithm has amortized update time $O(\sqrt{n})$.)
Using the flipping game
we get a local algorithm with amortized update time of $O(\sqrt{\log n})$ for low arboricity graphs.

The fastest local deterministic data structure for supporting adjacency queries requires a logarithmic query time, again even for dynamic forests.
Using the flipping game
we   get a deterministic local data structure for  adjacency queries supporting  queries and updates in $O(\log \log n)$
amortized time in low arboricity graphs, providing an exponential improvement over the
state-of-the-art. 


To prove these bounds, we upper bound the number of flips made by the flipping game in terms of the number of flips
made by the algorithm of BF for maintaining a $\Delta$-orientation.
We remark that the flipping game can be easily and efficiently distributed. 
This gives rise to a local distributed algorithm for
maintaining a maximal matching in a distributed network of low arboricity,
with amortized update time and message complexities of $O(\sqrt{\log n})$.
(Applying the distributed algorithm of \cite{PPS16} instead of the flipping game yields a global algorithm with  amortized message complexity  $\Omega(n)$.)

\section{Efficient Representations for Sparse Networks} \label{sec2}
\subsection{Low outdegree orientations with low local memory usage} \label{sec21}
Let $\Delta$ denote the outdegree threshold in the BF algorithm.
We present here a new algorithm for maintaining a $\Delta$-orientation in dynamic graphs of bounded arboricity $\alpha$.
Our algorithm achieves the same (up to a constant factor) parameters as the BF algorithm, yet it guarantees that the outdegree of all vertices
is bounded by the required threshold (i.e., $\Delta+1$) at all times.
We first (Section \ref{sec211}) describe the algorithm in a centralized setting, and then (Section \ref{sec212}) present a
distributed implementation.
Finally, we complement these results (Section \ref{app:blow}) by showing that the BF algorithm indeed blows up the outdegree of vertices, even after applying to it several natural adjustments.
\subsubsection{A new centralized algorithm that controls the outdegrees} \label{sec211}
Our algorithm handles edge deletions and insertions in the same way as the BF algorithm, until the outdegree of some vertex $u$ exceeds $\Delta$.
At this stage our algorithm does not apply a reset cascade, but rather aims at finding all the vertices that would ``benefit'' from
flipping their edges (in terms of reducing the value of a global potential function, details follow),
and then applies a cascade of \emph{anti-resets}, where vertices of sufficiently low outdegree
flip their incoming edges to be outgoing of them (rather than the other way around. as in the BF algorithm).  

Specifically, the algorithm starts by exploring the directed neighborhood ${N_u}$ outgoing of $u$, stopping at vertices of outdegree at most $\Delta' = \Delta - 2\alpha$.
That is, for each vertex of outdegree greater than $\Delta'$ that we reach starting from $u$, hereafter an \emph{internal vertex},
we explore all its out-neighbors. For each vertex of outdegree at most $\Delta'$ that we reach, hereafter a \emph{boundary vertex}, we do not do anything.
(Thus internal vertices of $N_u$ have outdegree greater than $\Delta'$ and all their out-neighbors belong to ${N_u}$,
whereas boundary vertices have outdegree at most $\Delta'$ and their out-neighbors may belong to ${N_u}$ due to other internal vertices,
but not due to the boundary vertices themselves.)
Denote by $I_u$ and $B_u$ the sets of internal and boundary vertices of ${N_u}$, respectively.
The algorithm constructs the digraph $\overrightarrow{G_u} = ({N_u},\overrightarrow{E_u})$, where $\overrightarrow{E_u}$
consists of all the outgoing edges of the internal vertices of ${N_u}$. This can be carried out in time linear in the size of $\overrightarrow{G_u}$.
Having constructed the digraph $\overrightarrow{G_u} = ({N_u},\overrightarrow{E_u})$, the algorithm proceeds by
computing a new orientation of $\overrightarrow{G_u}$ in which the outdegree of all vertices is bounded by
$2\alpha$ as follows. Initially we color (i.e., mark) all edges of $\overrightarrow{G_u}$.
Observe that at least one vertex of $\overrightarrow{G_u}$ is adjacent to at most $2\alpha$ colored edges;
we maintain a list $L_{2\alpha}$ of all vertices adjacent to at most $2\alpha$ colored edges.
We pick an arbitrary vertex in $L_{2\alpha}$, perform an anti-reset on it (flipping all its incoming edges to be outgoing of it), and then uncolor all its at most $2\alpha$ adjacent colored edges and update $L_{2\alpha}$ accordingly.
This process is repeated until no edge of $\overrightarrow{G_u}$ is colored, at which stage we have a valid $2\alpha$-orientation for $\overrightarrow{G_u}$.
Note that until a vertex performs an anti-reset, its outdegree may only decrease. Whenever a vertex performs an anti-reset,
its outdegree may increase, but to at most $2\alpha$, which means that a vertex never increases its outdegree beyond the maximum between $2\alpha$ and its initial out-degree.

Since each boundary vertex had at most $\Delta'$ out-neighbors in the entire graph, its new outdegree will be at most $\Delta' + 2\alpha = \Delta$,
and this also bounds its outdegree at any time during the process.
Moreover, since all outgoing edges of each internal vertex of $N_u$ are taken to $\overrightarrow{G_u}$,
the outdegree of each internal vertex never exceeds $\Delta+1$.
This process of computing a valid $2\alpha$-orientation while never blowing up the outdegree,
hereafter the \emph{anti-reset cascade procedure},  is inspired by the the static algorithm of \cite{AMZ97},
with the inherent difference that it works on a carefully chosen (possibly small) subgraph $\overrightarrow{G_u}$, whereas the reset cascade procedure
underlying the BF algorithm does not work on a precomputed subgraph, but rather on a subgraph that grows ``on the fly'' with the resets.
While it is easy to see that our procedure runs in linear time on any chosen subgraph (as with the BF algorithm),
the challenge is to show that the total cost of these procedures over all chosen subgraphs throuhgout the execution of our algorithm
is aymptotically the same as that of the BF algorithm.

\begin{lemma} \label{timeviaflips}
The total runtime of our algorithm is linear in the total number of edge flips made, assuming $\Delta \ge 5\alpha$.
\end{lemma}
\begin{proof}
Edge insertions and deletions are handled in constant time, until the outdegree of some vertex $u$ exceeds $\Delta$.
At this stage a digraph $\overrightarrow{G_u}$ as described above is constructed, along with the aforementioned list $L_{2\alpha}$, within time linear in the size of $\overrightarrow{G_u}$.
Then edges of $\overrightarrow{G_u}$ are flipped by the anti-reset cascade procedure, so that
each edge is flipped at most once.
By maintaining the list $L_{2\alpha}$ throughout the anti-reset cascade procedure, we can easily implement this procedure in time linear in the size of $\overrightarrow{G_u}$.
Note also that the size of $\overrightarrow{G_u}$ is given by the sum of outdegrees over the internal vertices of $N_u$.
To complete the proof, we argue that a constant fraction of the outgoing edges of each internal vertex of $N_u$ are flipped
during the anti-reset cascade procedure.
To see this, note that the outdegree of each internal vertex of $N_u$ reduces during this procedure from more than $\Delta' = \Delta - 2\alpha$
to at most $2\alpha$. Recalling that the outdegrees of vertices are bounded by $\Delta+1$ at all times,
at least $\Delta + 1-4\alpha$ out of at most $\Delta + 1$ outgoing edges (which is at least a $\frac{1}{5}$-fraction assuming  $\Delta \ge 5\alpha$)
of each internal vertex must have been flipped during the procedure.
\QED
\end{proof}

Although our algorithm and the BF algorithm are inherently different,
we use a potential function argument similar to the one in \cite{BF99} to bound the number of flips made by our algorithm,
which by Lemma \ref{timeviaflips} also bounds its total runtime (up to a constant factor).
The key insight is that we can apply a potential function argument \emph{globally}, i.e., for all the anti-resets together, rather than to each one of them separately as
was done for resets by \cite{BF99}.

Suppose one can maintain a $\delta$-orientation for some sequence of $t$ edge updates while doing $f$ edge flips, starting with the empty graph.
As in \cite{BF99}, we define an edge to be  \emph{good} if its orientation in our algorithm is the same as in the $\delta$-orientation and \emph{bad} otherwise.
We define the potential $\Psi$ to be  the number of bad edges in the current graph.
Initially $\Psi = 0$.
Each insertion or a flip performed by the $\delta$-orientation increases $\Psi$ by at most one, while edge deletions may only decrease $\Psi$.
All edge flips made by our algorithm are due to the anti-reset cascade procedures.
Consider some digraph $\overrightarrow{G_u}$ on which an anti-reset cascade procedure is applied throughout the execution of our algorithm,
and note that all the edges of $\overrightarrow{G_u}$ are outgoing of internal vertices of $N_u$ before the procedure starts.
Let $v$ be an arbitrary internal vertex of $N_u$, and note that its outdegree before the procedure starts is greater than $\Delta'$.
Moreover, by the definition of a $\delta$-orientation, at most $\delta$ of $v$'s outgoing edges at that moment are good.
As a result of the procedure,  these $\delta$ edges may become bad.
However, since $v$'s outdegree reduces to at most $2\alpha$ at the end,
at least $\Delta' + 1 - 2\alpha -  \delta$ edges were bad and become good.
It follows that $\Psi$ is decreased by at least $\Delta' + 1 - 2\alpha - 2\delta$ per each internal vertex.
Consequently, the total number of vertices that serve as internal vertices of some digraph $\overrightarrow{G_u}$ throughout the execution of our algorithm
is at most $(t + f) / (\Delta' + 1 - 2\alpha - 2\delta)$.
Since the outdegree of all vertices is bounded by $\Delta + 1$ at all times,
the total number of edge flips made by our algorithm is bounded by  $(t + f)(\Delta + 1) / (\Delta' + 1 - 2\alpha -  2\delta)$.
Assuming $\Delta \ge 6\alpha + 3\delta$,
it follows that $(t + f)(\Delta + 1) / (\Delta' + 1 - 2\alpha - 2\delta) \le 3(t+f)$.
\subsubsection{A distributed implementation with low local memory usage} \label{sec212}
Consider a vertex $u$ whose outdegree exceeds $\Delta$.
The centralized algorithm starts by exploring the directed neighborhood ${N_u}$ and coloring all edges of the digraph $\overrightarrow{G_u} = ({N_u},\overrightarrow{E_u})$ as described above.
We can distribute this step using broadcast and convergecast in a straightforward way.
However, we also need to make sure that the local memory usage at processors is bounded by $O(\Delta)$.
To this end, every internal processor (with outdegree larger than $\Delta'$) will be responsible for coloring its outgoing edges.
Throughout this broadcast we also compute the directed BFS tree $T_u$ on $N_u$, so that each processor will hold information about its parent in $T_u$,
using which we can easily carry out the subsequent convergecast.
The number of rounds will be linear in the depth $h$ of $T_u$, whereas the number of messages will be linear in the size of $\overrightarrow{G_u}$.

The centralized algorithm continues by running the anti-reset cascade procedure.
This procedure is inspired by the static algorithm of \cite{AMZ97}, for which an efficient distributed implementation was given in \cite{BE10}.
We cannot use the distributed algorithm of \cite{BE10}, however, since it lets processors communicate with all their neighbors,
hence the local memory usage will depend on the maximum degree in the network, which can be significantly larger than $O(\Delta)$.
(Recall that here $\Delta$ stands for the out-degree threshold, which is linear in the arboricity $\alpha$, and may be $n/\alpha$ times smaller than the maximum degree.)

The distributed algorithm that we propose is a variant of \cite{BE10}, and works as follows.
First, we  change the threshold $\Delta'$ of the centralized algorithm from $\Delta - 2\alpha$ to $\Delta - 5\alpha$.
To compensate for the decrease in the value of $\Delta'$, we increase $\Delta$ by a constant factor.
(By letting $\Delta$ increase by a constant factor, the above potential function argument will carry over smoothly.)
In each round $i = 1,2,\ldots,\log |N_u|$, all the colored processors send messages on each of their colored outgoing edges.
Every colored processor that receives at least one message checks if the number of its colored outgoing edges plus the number of messages it received is bounded by $5\alpha$.
If so, it flips all the edges along which it received messages to be outgoing of it, and then uncolors itself and all its outgoing edges.

This distributed anti-reset cascade procedure implicitly assumes that all processors of
$\overrightarrow{G_u}$ wake up simultaneously, and the entire subgraph $\overrightarrow{G_u}$ (both edges and processors) is colored at this moment.
To justify this assumption, before initiating this procedure, we perform a broadcast along $T_u$, in which each processor at directed distance $i$ from the root receives message $h-i$.
A processor receiving message $h-i$ will wake up in exactly $h-i$ rounds from the time it received the message to
color itself and its outgoing edges, and then participate in the distributed anti-reset cascade procedure.

We next analyze this procedure.
In each round $i = 1,2,\ldots,\log |N_u|$, at least 3/5 of the colored processors are adjacent to at most $5\alpha$ colored edges,
since the subgraph  induced by the colored edges has arboricity at most $\alpha$.
This means that the number of colored vertices reduces by a factor of $5/2 > 2$ in each round, hence after the last round all edges have been uncolored,
and we obtain a $5\alpha$-orientation for $\overrightarrow{G_u}$.
Moreover, we argue that the number of edges being uncolored in each round is no smaller than the number of edges that remain colored.
To see this, fix an arbitrary round $i$, consider the graph $G_i$ induced by the colored edges at the beginning of the round,
and denote by $V'_i$ and $V''_i$ the set of vertices that get uncolored and remain colored at the end of round $i$, respectively.
Since no vertex in $V''_i$ get uncolored in round $i$, the degree of each vertex of $V''_i$ is at least $5\alpha$ in $G_i$.
However, the subgraph $G''_i$ of $G_i$ induced by the vertex set $V''_i$ has arboricity at most $\alpha$, hence at least half of the vertices of $V''_i$
have at most $4\alpha$ neighbors in $V''_i$, which means their remaining $\ge \alpha$ neighbors are in $V'_i$.
The assertion now follows since the number of edges in $G''_i$, or the number of edges that remain colored, is at most $\alpha \cdot |V''_i|$,
whereas the number of edges that got uncolored is at least $\alpha \cdot |V''_i|$.
Consequently, the number of messages sent in each round decays geometrically, hence the total number of messages sent is linear in the size of $\overrightarrow{G_u}$.
Note also that this procedure terminates within $\log |N_u|$ rounds, which does not exceed the number of messages sent.

\begin{theorem} \label{bas}
For any $\alpha \ge 1$ and $\Delta = \Omega(\alpha)$ and any arboricity $\alpha$ preserving sequence of edge and vertex updates starting from   empty graph,
there is a distributed algorithm for maintaining a $\Delta$-orientation (in the $\congest$ model) with an optimal (up to a constant) amortized message complexity,
and the same (or better) amortized update time.  
The local memory usage at all vertices is $O(\Delta)$ at all times, which is also optimal.
For $\Delta = O(\alpha)$, we obtain $O(\alpha)$-orientation with $O(\log n)$ amortized update time and message complexities, with $O(\alpha)$ local memory usage.
\end{theorem}
The worst-case update time of the above algorithm 
may be high. The bottleneck is the time needed to explore the directed neighborhood ${N_u}$ and compute the tree $T_u$ on which the broadcast and convergecast are carried out,
which is linear in the depth of $T_u$.
To remedy this, we show that the aforementioned potential function argument will continue to work if we truncate
the tree at a carefully chosen depth parameter $O(\log n)$, thereby reducing the worst-case update time to $O(\log n)$.
This truncation, however, is nontrivial. In particular, we do not truncate $T_u$ at depth $\Theta(\log n)$, but rather at the minimal depth $i$ for which the number of vertices is smaller than $O(\Delta)^i$, where the constant hiding in the $O$-notation should be chosen with care.
We omit these details, since our focus in this work is on amortized rather than worst-case bounds.

\subsubsection{Outdegree blowup in the BF algorithm} \label{app:blow}
\begin{lemma}
 For graphs with arboricity 1 (i.e., for forests), the original BF algorithm does not increase the outdegree of a vertex beyond $\Delta + 1$ during
a reset cascade that follows an edge insertion.
 \end{lemma}
\begin{proof}
Note that the graph is a forest, not necessarily a tree. However, as the reset cascade does not reset vertices outside the subtree containing $r$,
we may henceforth restrict our attention to that subtree, denoted $T$.
Let $\overrightarrow{T}$ be the oriented tree \emph{before the cascade started} and let $r$ be the vertex that we reset first in the cascade.
(So in $\overrightarrow{T}$, the outdegree of $r$ is $\Delta + 1$ and the outdegrees of all other vertices is $\le \Delta$.)

\begin{observation} \label{dipath}
If the cascade resets $v$ then there is a directed  path from $r$ to $v$ in $\overrightarrow{T}$.
\end{observation}
\vspace{-5pt}
We prove this observation by induction on the position of the reset in the reset sequence of the cascade.
For the basis $v=r$, and the statement holds vacuously. For the induction step, consider a reset of an arbitrary vertex $v\not= r$,
and suppose that the statement holds for any preceding reset in the reset sequence of the cascade.
Note that $v$'s outdegree at the time of the reset is larger than $\Delta$. So when the reset occurs $v$ must have an outneighbor, say $w$,  that was not an outneighbor of $v$ in $\overrightarrow{T}$.
Since the orientation of edge $(v,w)$ flips only due to a reset, there must have been at least one reset on $w$ preceding the reset on
$v$ in the reset sequence. By induction there is a directed  path from $r$ to $w$ in $\overrightarrow{T}$. Furthermore, the edge $(v,w)$ was oriented from
$w$ to $v$ in $\overrightarrow{T}$. Hence there is a directed path from $r$ to $v$ in $\overrightarrow{T}$, as required.

Now we prove the lemma by contradiction. Consider the time during the reset cascade in which the outdegree of
a vertex $v\not=r$ becomes $\Delta + 2$. Then at this time vertex $v$ must have two outneighbors $w_1$ and $w_2$ which were not
outneighbors of $v$ in $\overrightarrow{T}$. It follows that there must have been a reset on $w_1$ and on $w_2$.
By the observation above there are directed paths in $\overrightarrow{T}$ from $r$ to $w_1$ and from $r$ to $w_2$.
This means that there are two directed paths in $\overrightarrow{T}$ from $r$ to $v$, one ending with the arc $(w_1,v)$ and another ending with the arc
$(w_2,v)$, contradicting the fact that the arboricity is 1.

If the outdegree of $r$ becomes $\Delta + 2$, then $r$ has an outneighbor $w$ that was not an outneighbor of $r$ at
$\overrightarrow{T}$. As before, there must have been a reset on $w$, so by Observation \ref{dipath} there is a directed path from $r$ to
$w$ in $\overrightarrow{T}$. This path together with the arc $(w,r)$ closes a direccted cycle in $\overrightarrow{T}$, a contradiction.
\QED
\end{proof}

The following lemma shows that when the arboricity is larger than $1$ we may get vertices with very large outdegree during the reset cascade process.

\begin{lemma} \label{basiclb}
There exists a graph with arboricity 2, for which the original BF algorithm may increase the outdegree of a vertex to $\Omega(n / \Delta)$.
\end{lemma}
\begin{proof}
Consider an ``almost perfect'' $\Delta$-ary tree oriented towards the leaves.
Specifically, the only difference from a perfect $\Delta$-ary tree is that each of the parents of the leaves has $\Delta - 1$ children
rather than $\Delta$, but it also has an outgoing edge to some vertex $v^*$. So the arboricity of the graph is 2.

Suppose that the outdegree of the root increases to $\Delta + 1$ due to some edge insertion, thus starting a reset cascade.
When the parents of the leaves are reached, they will have outdegree of $\Delta + 1$.
Hence they will be reset one after another,
which gradually increases the outdegree of $v^*$ from 0 to $\Omega(n / \Delta)$.
\QED
\end{proof}
{\bf Remark.}
The lower bound $\Omega(n /\Delta)$ on the maximum outdegree provided by Lemma \ref{basiclb} is tight.
To see this, note that only vertices with degree greater than $\Delta$ may perform resets.
In a graph of arboricity $\alpha$, there are at most $2\alpha (n/\Delta)$ such vertices,
implying that the outdegree of a vertex will not increase by more than $2\alpha (n/\Delta)$ during the reset cascade.
\vspace{8pt}
\\
{\bf Largest outdegree first.~}
There is a natural adjustment to the reset cascade one can make in order to control the outdegree blowup during the cascade,
specifically, to reset vertices of larger outdegree first. This is easily achieved with $O(1)$ overhead on each operation of the cascade, by keeping
the vertices whose outdegree  is larger than $\Delta$ in a heap $H$, using the outdegree of a vertex as its key.
We need to be able to extract the maximum element in $H$ when we decide on the next vertex to reset, and to increase the key of a vertex by $1$ when we flip an edge.
It is straightforward to implement such an heap so that each operation takes $O(1)$ time.
The following lemma shows that this adjustment suffices to control the outdegree from blowing up by more than a logarithmic factor.
We remark that the proof of this lemma is similar to the proofs of Lemma 6 and 7 of \cite{HTZ14}.
\begin{lemma} \label{lem:heu1} 
If we always reset a vertex of largest outdegree first, then the outdegree
of a vertex never exceeds $4\alpha \lceil \log(n / \alpha) \rceil + \Delta$.
\end{lemma}
\begin{proof}
To prove Lemma \ref{lem:heu1}, we employ the following two claims.
\begin{claim} \label{claim:1}
A vertex $v$ that has outdegree $\Delta^* > \Delta$ during the cascade has distinct neighbors
$v_{\Delta+1},\ldots,v_{\Delta^*}$ where the outdegree of $v_i$ during the cascade is at least $i-1$.
\end{claim}
\begin{proof}
Focus on an arbitrary vertex $v$ and consider a maximal subsequence of the reset cascade in which the outdegree of
$v$ does not decrease. At the beginning of this subsequence, vertex $v$ has  outdegree $\le \Delta$ ($v$ has  outdegree $0$ if it was reset
just before the subsequence starts, and outdegree $\le \Delta$ if the subsequence starts with the first reset of the cascade).
By the largest-reset adjustment, the
 outdegree of $v$ increases from $i$ to $i+1$ due to a reset on a neighbor $v_i$ of outdegree $\ge i$.
Clearly $v_i\not= v_j$, so the claim follows.
\QED
\end{proof}

\begin{claim} \label{claim:2}
Let $v$ be a vertex of outdegree $\Delta^* \ge 4\alpha k + \Delta$ during the cascade.
Then for every $t$, $1\le t\le k$, there are $\ge 2^t\alpha$ vertices at distance $\le t$ from $v$
whose outdegree during the cascade is $\ge 4\alpha(k-t) + \Delta$.
\end{claim}
\begin{proof}
The proof is by induction on $t$. The basis $t=1$ follows   from Claim \ref{claim:1}.
For the induction step, we assume the statement holds for some $t < k$, and prove it for $t+1$.
Let $V_t$ the set of vertices at distance  $\le t$ from $v$
whose outdegree during the cascade is $\ge 4\alpha(k-t) + \Delta$.
By induction $|V_t|\ge 2^t\alpha$.
By Claim \ref{claim:1}, each $v\in V_t$ has $4\alpha$ neighbors whose outdegree (and degree)
during the cascade is $\ge 4\alpha(k-t-1) + \Delta$. 
Let $V'_{t}$ be the set of all these neighbors, and note that all vertices in $V_{t+1}=V_t\cup V_t'$
are at distance $\le t+1$ from $v$ and their outdegree during the cascade is $\ge 4\alpha(k-t-1) + \Delta$.
Moreover, the number of edges in the graph induced by $V_{t+1}$ is $\ge |V_t|\frac{4\alpha}{2}\ge 2^{t+1}\alpha^2$.
Since the arboricity of the graph induced by $V_{t+1}$ is at most $\alpha$, it follows that this graph must have $\ge 2^{t+1}\alpha$ vertices, which completes the induction step.
\QED
\end{proof}

We conclude that the outdegree of a vertex $v$ cannot exceed $4\alpha \lceil \log(n / \alpha) \rceil + \Delta$,
as otherwise there would be more than $n$ vertices in the graph by Claim \ref{claim:2}.
This completes the proof of Lemma \ref{lem:heu1}. 
\end{proof}


We next show that the upper bound  of Lemma \ref{lem:heu1} is tight for the BF algorithm with the above adjustment.
Our lower bound holds even if we make another natural adjustment to the   algorithm, 
where we orient a newly inserted edge from the vertex with lower outdegree to the vertex with higher outdegree.

For every $i\ge 2$, we define a directed graph $G_i$ on $2^i$ vertices, in which each vertex has outdegree $2$, except for two special
vertices that have
ourdegree $0$.
The graphs $G_2$ and $G_3$ are shown in Figure \ref{fig1}. The graph $G_2$ consists of two vertices, denoted by $a$ and $b$, and a cycle of length $2$ which we denote by
$C_1$.
\begin{figure}[h!]
\begin{center}
\hspace*{-3cm}
\includegraphics[scale=0.5]{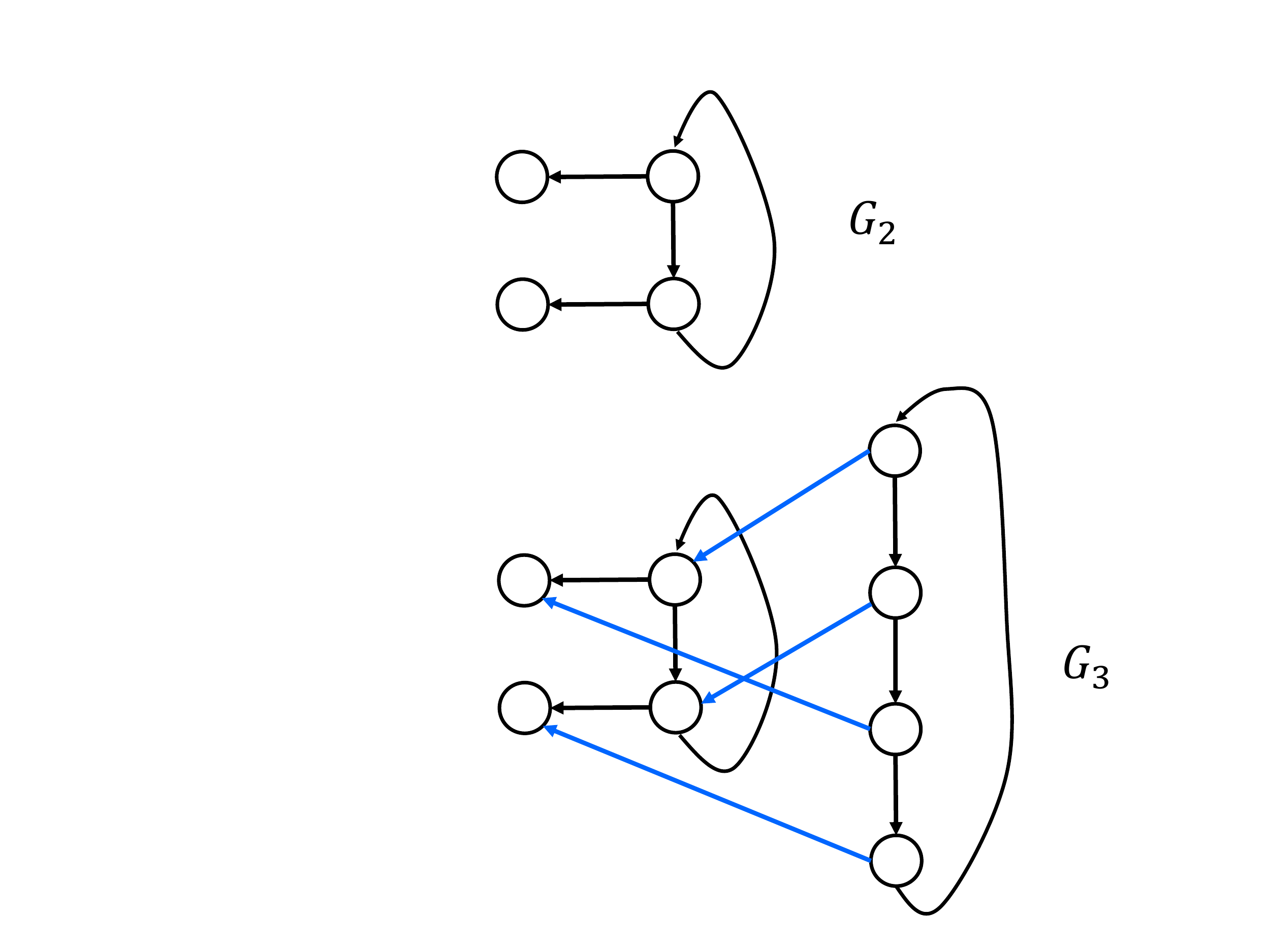}
\end{center}
\caption{The graphs $G_2$ and $G_3$ \label{fig1}}
\end{figure}

In general we obtain $G_{i+1}$ from $G_i$ by adding to $G_i$ a cycle $C_i$ on $2^{i}$ vertices and an outgoing edge from
each vertex of $C_i$ to a unique (but arbitrary) vertex of $G_i$, such that each vertex of $G_i$ is connected  in $G_{i+1}$ to a single vertex of
$C_i$. The proofs of the following observation and lemmas are immediate.

\begin{observation}
For any $i\ge 2$ the
 graph $G_i$ has $2^i$ vertices. Each vertex of $G_i$ has outdegree $2$ except for the vertices $a$ and $b$ of $G_2$ that have outdegree
$0$.
\end{observation}

\begin{lemma}  \label{forapp2}
The arboricity of $G_i$ is $2$.
\end{lemma}
\begin{proof}
By induction on $i$. We can easily decompose $G_2$ into two forests.
Assuming we can decompose $G_i$ into two forests $F_1$ and $F_2$, we decompose $G_{i+1}$ into two forests as follows.
We index the vertices on $C_i$ from $1$ to $2^i$,
and add to $F_1$ (respectively, $F_2$) the two outgoing edges of every vertex of odd (resp., even) index.
It is easy to verify that $F_1$ and $F_2$ are cycle-free.
\QED
\end{proof}

\begin{lemma}  \label{forapp3}
We can construct $G_i$ starting from an empty graph on $2^i$ vertices by inserting the edges one after another,
such that each edge is oriented from the vertex of lower outdegree to the vertex of higher outdegree at the time of its insertion.
\end{lemma}
\begin{proof}
By induction on $i$. To construct $G_{i+1}$ we first add the edges of $G_i$, then the edges from $C_i$ to the vertices of
$G_i$ and last the edges between the vertices of $C_i$. It is easy to verify that the orientations are assigned properly if newly inserted edges
are oriented towards the higher outdegree endpoint.
\QED
\end{proof}


Assume for simplicity that $\Delta = 2$ and consider the reset cascade that starts when we add to some vertex $v$ of $G_i$ an outgoing edge such that
its outdegree  increases to $3$.
(This edge to be oriented out of $v$ should be incident to a vertex whose outdegree is not smaller than the outdegree of $v$ and is external to $G^I$.)
Flipping $v$ increases the outdegree of the vertex $v'$ following $v$ on $C_{i-1}$, as well as the outdegree of some vertex in $G_{i-1}$ connected to $v$. So the next flip may be on $v'$. We continue this way flipping all vertices of $C_{i-1}$
while increasing the outdegree of all vertices of $G_{i-1}$ from $2$ to $3$, except for vertices $a$ and $b$ of $G_2$ whose
outdegree increases from $0$ to $1$.
Next we flip the vertices of $C_{i-2}$ and so on. Right before flipping the vertices of $C_1$ they have outdegree $i$.
The following lemma specifies the invariant being maintained during the cascade.
Its proof is straightforward by induction on the operations of the cascade.
\begin{lemma} \label{lem:invariant}
When we flip the vertices of $C_j$ for some $j \le i-1$, the outdegrees of   vertices are as follows:
~~~(1)
Vertices of $C_\ell$ for $j < \ell \le i-1$ have outdegree $\le 3$.
~~~(2)
Vertices of $C_j$ that were already flipped have outdegree $\le 2$.
~~~(3)
A vertex of $C_\ell$ for $\ell < j$ that is incident to a vertex of $C_j$ that was already   flipped has
outdegree $1+i-\ell$ and a vertex of $C_\ell$ for $\ell < j$ that is incident to a vertex of $C_j$ that was not flipped already has
outdegree $i-\ell$.
\end{lemma}

By applying  the  Invariant   of Lemma \ref{lem:invariant} to the point when we finished flipping the vertices of
$C_2$,
 it follows that
during  a cascade on $G^i$ that starts by increasing the outdegree of a vertex of $C_{i-1}$, we get that the vertices of
$C_1$
have outdegree $i$ right before they are flipped.
We derive the following corollary.
\begin{corollary}
The BF algorithm with the two adjustments above may blowup the outdegree of a vertex to $\log n $ during an insertion into
a graph with $O(n)$ vertices. (In fact $n+1$ vertices suffice.)
\end{corollary}

If the threshold of the BF algorithm is some $\Delta > 2$ then we can adapt the example described above by adding to each vertex
$\Delta - 2$ ``private'' neighbors. This increases the number of vertices to $n'= n(\Delta-1)+1$.
The maximum outdegree reached during the reset cascade is $\log(n'/(\Delta-2)) + \Delta-3$,
hence this lower bound matches the upper bound of Lemma \ref{lem:heu1} up to a constant factor,
for graphs of constant arboricity.


\medskip

Next, we  generalize the construction to show that the BF algorithm with the two adjustments above may blowup the outdegree of
a vertex to $\Omega(\alpha \log (n/\alpha))$ during a reset cascade initiated by an edge insertion in a graph with arboricity $\alpha$ and $n$ vertices.



We describe the construction in two stages. First we need to change the graph $G_i$ slightly for technical reasons,
and then we construct a graph $G^\alpha_i$ on which we demonstrate the reset cascade.

The technical change of $G_i$ is as follows.

\begin{enumerate}
\item We change $G_2$ to the graph in Figure \ref{fig2}.
\item When we construct $G_{i+1}$ from $G_i$ we make the cycle $C_{i}$ of length $|V_i|+1$ (rather than $|V_i|$), where $V_i$ is the vertex set of $G_i$.
One special vertex of $C_i$ is not connected to any vertex of $G_i$ in $G_{i+1}$. We denote this special vertex by $s_i$.
\end{enumerate}
\begin{figure}[h!]
\begin{center}
\hspace*{-3cm}
\includegraphics[scale=0.5]{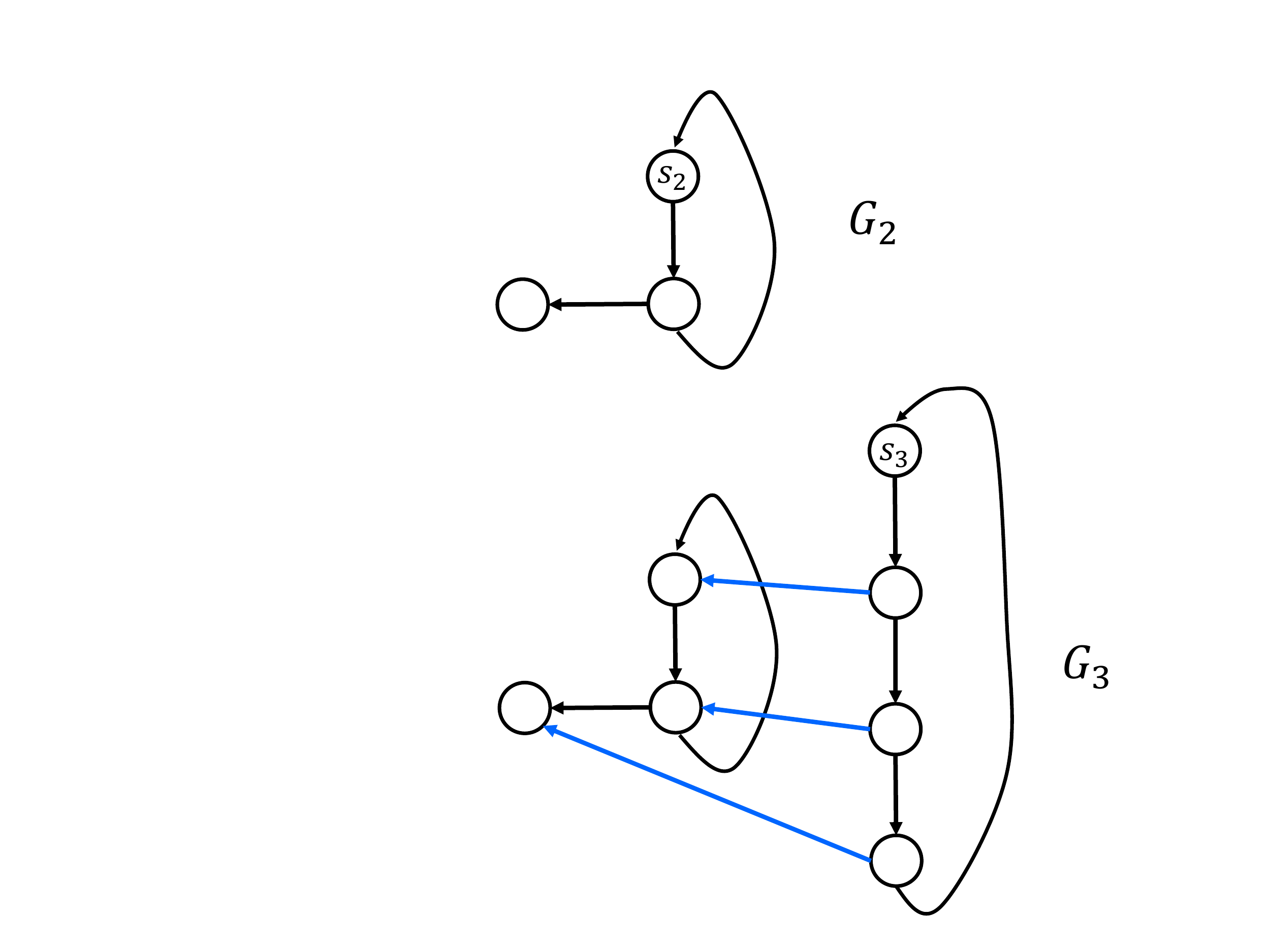}
\end{center}
\caption{The graphs $G_2$ and $G_3$ in the generalized construction, before we replace each vertex by $\alpha$ vertices \label{fig2}}
\end{figure}

The graph $G^\alpha_i$ is constructed from $G_i$ by
performing the following modification for every $2\le i\le \alpha$.

\begin{enumerate}
\item For each vertex vertex $u \in G_i$ we have  $\alpha$ vertices, $u^1,\ldots, u^\alpha$ in $G^\alpha_i$.
\item For each edge from a vertex $u$ of $C_i$ to a vertex $v$ of $G_i$, put a complete bipartite clique between the vertices
$u^1,\ldots, u^\alpha$ and $v^1,\ldots, v^\alpha$ in $G^\alpha_i$. Each edge $(u_i^j,v_i^\ell)$ is directed from $u_i^j$ to $v_i^\ell$.
\item For each edge from a vertex $u$ of $C_i$ to the next vertex $v$ of $C_i$, put a complete bipartite clique between the vertices
$u^1,\ldots, u^\alpha$ and $v^1,\ldots, v^\alpha$. Each edge $(u_i^j,v_i^\ell)$ is directed from $u_i^j$ to $v_i^\ell$.
\item Connect the vertices $s^1_i,\ldots, s^\alpha_i$ to another set of vertices $t^1_i,\ldots, t^\alpha_i$.
Make $s^1_i,\ldots, s^\alpha_i$ a clique and orient it such that an edge $(s_i^j, s_i^\ell)$ for $j<\ell$ is oriented from
$s_i^j$ to $s_i^\ell$. Make $t^1_i,\ldots, t^\alpha_i$ a clique and orient it analogously. Connect
$s_i^j$ to $t_j^\ell$ for $\ell \le j$. Notice that the number of edges that are directed from $s_i^j$ to one of
$s^1_i,\ldots, s^\alpha_i$ and $t^1_i,\ldots, t^\alpha_i$ is exactly $\alpha$. See Figure \ref{fig3}.
\end{enumerate}

The analysis of this generalization is analogous to the analysis of the construction for $\alpha=2$, and thus omitted from this extended abstract.

\begin{figure}[h!]
\begin{center}
\hspace*{-0.7cm}
\includegraphics[scale=0.4]{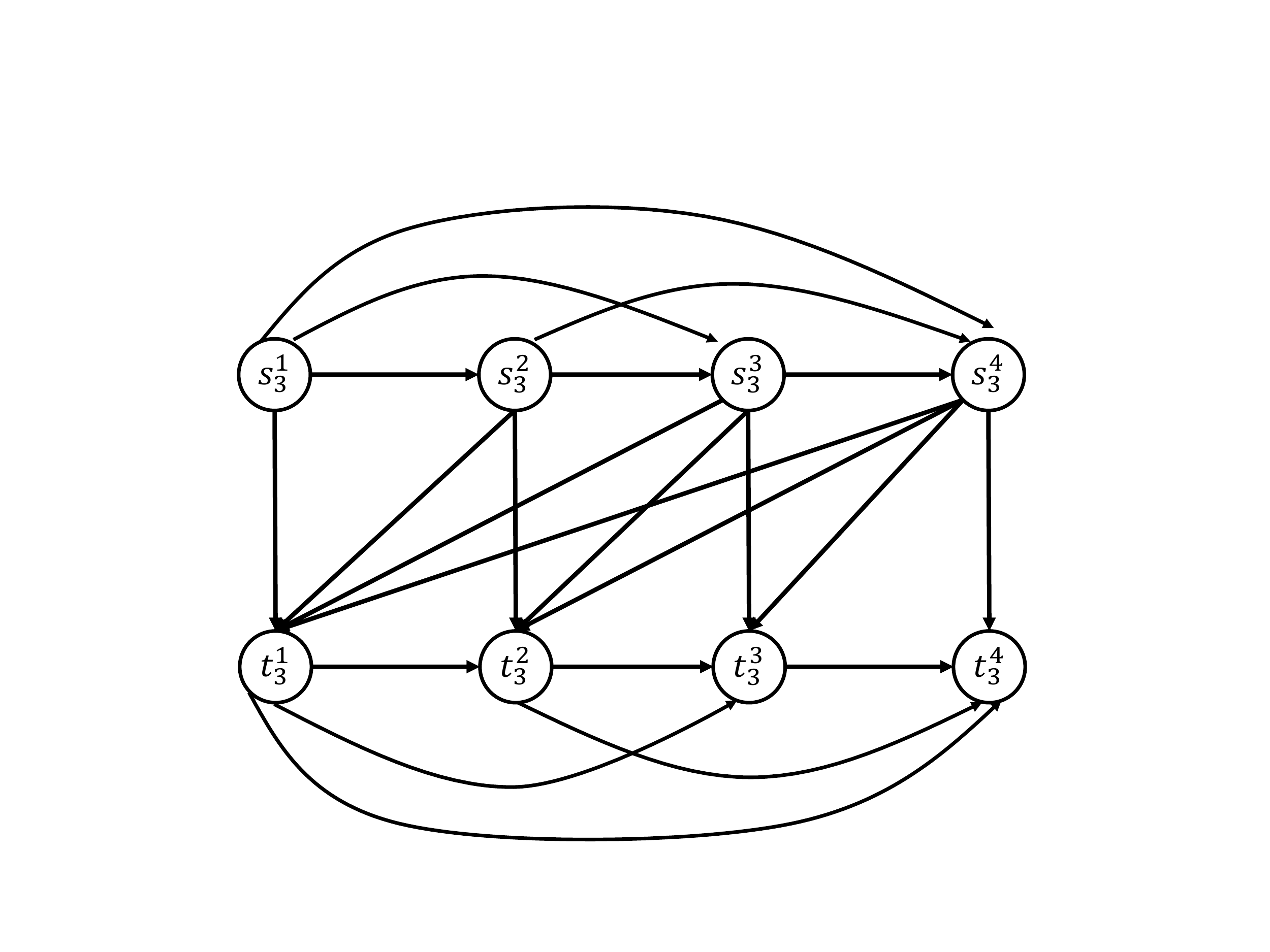}
\end{center}
\vspace{-0.75cm}
\caption{The graph which we use to replace $s_3$ for $\alpha = 4$.
\label{fig3}}
\end{figure}

\subsection{Efficient representations of sparse distributed networks, with applications} \label{sec22}
In this section we describe a natural representation of sparse distributed networks, along with some applications.

\subsubsection{Forest decomposition and adjacency queries~}
For a distributed network with arboricity $\alpha$,
Theorem \ref{bas} provides a distributed algorithm (in the $\congest$ model) for maintaining a low outdegree orientation with low local memory usage.
Such an orientation can be viewed as a representation of the network, and it finds two natural applications.
First, due to the equivalence between the edge orientation and the forest decomposition problems shown in \cite{PPS16},
we obtain a distributed algorithm for maintaining a decomposition into $O(\Delta)$ forests within an optimal (up to a constant) amortized message complexity,
and the same (or better) amortized update time, with $O(\Delta)$ local space, for any $\alpha \ge 1$ and $\Delta = \Omega(\alpha)$.

We can then use this forest decomposition to maintain efficient distributed \emph{adjacency labeling schemes}.
An adjacency labeling scheme assigns an (ideally short) label to each vertex, allowing one to infer if any two vertices $u$ and $v$ are neighbors directly from their labels.
For an adjacency representation
scheme to be useful, it should be capable of reflecting online the current up-to-date
picture in a dynamic setting. Moreover, the algorithm for generating and revising the labels must
be distributed.
Given an $f$-forest-decomposition for $G$, the label of each vertex $v$ can be given by $\Label(v)=(\id(v), \id(w_1), \ldots, \id(w_f))$ where $w_i$ is the parent of $v$ in the $i$th forest.
We derive the following result.
\begin{theorem}
For any $\alpha \ge 1$ and any arboricity $\alpha$ preserving sequence of   updates,
there is a distributed algorithm (in the $\congest$ model) for maintaining an   adjacency labeling scheme with label size of $O(\alpha \cdot \log n)$ bits
with $O(\log n)$ amortized message complexity and update time, with $O(\alpha)$ local memory usage. 
\end{theorem}

\subsubsection{A complete representation}
A low outdegree orientation may not quality as a complete representation of the network, since a processor cannot access its incoming neighbors,
and in particular it cannot communicate with them. Next, we describe a complete representation of a distributed network.

Consider a processor $v$ with $k$ incoming neighbors $v_1,\ldots,v_k$.
For each $i$, we will make sure that $v_i$ holds information on $v_{i-1}$ and $v_{i+1}$, with $v_0 = v_{k+1} = \NULL$,
and $v$ will hold information on an arbitrary processor among these, say $v_k$.
(This information that we hold per neighbor $v_i$ of $v$ should be enough for $v$ to communicate with $v_i$ directly.)
Since the network may change dynamically, we need to update the ``extra'' local information that we hold at processors efficiently.
We refer to the processors $v_1,\ldots,v_k$ as \emph{siblings}, and $v$ is referred to as their \emph{parent}.
For each $i, 1 \le i \le k$, $v_{i-1}$ and $v_{i+1}$ are referred to as the \emph{left sibling} and \emph{right sibling} of $v_i$, respectively.
(The left and right siblings of $v_1$ and $v_k$, respectively, are defined as $\NULL$.)
Note that each processor holds information on two of its siblings, per any parent.
Since the number of parents of any processor $v$ is given by its outdegree, the information regarding all siblings of $v$ over all of its parents is linear in its outdegree.
In addition, any processor $v$ holds information on a single incoming neighbor $v_k$, as described above.
Together with all its outgoing neighbors, the total information at a processor is linear in its outdegree.
Since the outdegree of the underlying edge orientation is (close to) linear in the arboricity of the network,
we can make sure that the local information at processors is (close to) linear in the arboricity, yielding the required bound on the local memory usage.

Following an insertion of edge $(u,v)$  that is oriented from $u$ to $v$, $u$ will hold information on $v$ (by the underlying edge orientation).
We also make sure that $v$ will hold information on $u$ by designating $u$ as $v_{k+1}$, i.e., $u = v_{k+1}$ takes the role of $v_k$.
Subsequently, $v$ sends a message with information on $v_{k+1}$ to $v_k$ and another message with information on $v_k$ to $v_{k+1}$,
so that $v_{k}$ (respectively, $v_{k+1}$) will hold information on $v_{k+1}$ (resp., $v_k$) as its new right (resp., left) sibling.
Following a deletion of edge $(u,v)$ that is oriented from $u$ to $v$, with $v$ being the parent of $u = v_i$ for some index $i$,
$v_i$ sends a message with information on both $v_{i-1}$ and $v_{i+1}$ to $v$. Subsequently, $v$ sends two messages (in parallel), one to $v_{i-1}$ and another to $v_{i+1}$,
informing $v_{i-1}$ (respectively, $v_{i+1}$) that its right (resp., left) sibling has changed from $v_i$ to $v_{i+1}$ (resp.,  $v_{i-1}$). Note that we send a message along the deleted edge $(v,v_i)$
in order to update the representation following an edge deletion, i.e., we support a \emph{graceful} edge deletion but not an \emph{abrupt} one.
(In the former, the deleted edge may be used for exchanging messages between its endpoints, and retires   only once the representation has been updated.
In the latter, while the endpoints of the deleted edge discover that the edge has retired, it cannot be used for any communication.)
A similar update is triggered by edge flips and vertex updates, where we only support a graceful deletion of vertices.
\vspace{8pt}
\\
{\bf Some applications.~}
The drawback of such a representation is that a processor cannot communicate with its in-neighbors \emph{in parallel}.
For $v$ to be able to send a message to an in-neighbor $v_i$, it first needs to retrieve the information on $v_i$ required for communicating with it.
To this end, $v$ has to \emph{sequentially} scan and communicate with all its in-neighbors $v_k,v_{k-1},\ldots,v_{i+1}$, starting at $v_k$ (on which $v$ holds information) and finishing at $v_i$.
For some applications, however, such a sequential scan of the in-neighbors is not needed.

For the sake of conciseness, in what follows we focus on edge updates and flips. Vertex updates can be handled in a similar way.

As a first application, consider the problem of maintaining a maximal matching in a distributed network that changes dynamically.
Instead of maintaining the information on the in-neighbors as described above,
we will maintain information only on the \emph{free} in-neighbors.
 More specifically, information on the free in-neighbors is being distributed among them in the manner described above.
Whenever a processor changes status from free to matched, or vice versa,
it notifies all its out-neighbors about that. (Recall that each processor has complete information on all its out-neighbors, and can communicate with all of them in parallel.
Interestingly, there is no need to exploit parallelism here.) Any processor that receives such information makes sure to update the relevant local information regarding its free in-neighbors, which is distributed into the relevant neighbors,
following along similar lines to the above.
The rest of the algorithm now proceeds as in the centralized setting \cite{NS13}.
Specifically, following an edge insertion, we match the two endpoints if they are free,
and otherwise there is nothing special to do (besides updating the underlying representation).
Following a deletion of an unmatched edge, there is again nothing special to do.
Finally, following a deletion of a matched edge $(u,v)$, $u$ and $v$ exchange messages with their out-neighbors,
attempting to find a free neighbor among them. Let us focus on $u$ ($v$ is handled in the same way).
If none of $u$'s out-neighbors is free, $u$ needs to check whether it has a free in-neighbor.
Since we made sure to (distributively) maintain information on the free in-neighbors of each vertex, including $u$,
and as there is no need to perform a sequential scan over these neighbors of $u$ (the first one, if any, will do),
we conclude that the amortized message complexity of the algorithm, and thus the amortized update time, is dominated (up to constant factors) by the
maximum among the outdegree bound of the underlying orientation and the cost of maintaining that orientation.
\begin{theorem} \label{basinet}
For any $\alpha \ge 1$ and any arboricity $\alpha$ preserving sequence of edge and vertex updates starting from an empty graph,
there is a distributed algorithm (in the $\congest$ model) for maintaining a maximal matching with an amortized update time
and message complexities of $O(\alpha + \log n)$.
The local memory usage is $O(\alpha)$.
\end{theorem}

As a broader application, we revisit the \emph{bounded degree sparsifiers} introduced recently in \cite{Sol18}.
Informally, a bounded degree $(1+\eps)$-sparsifier for a graph $G = (V,E)$, a degree parameter $\Delta$ and a slack parameter $\eps > 0$
is a \emph{subgraph} $H$ of $G$ with maximum degree bounded by $\Delta$ that preserves certain quantitative properties of the original graph up to a (multiplicative) factor of $1+\eps$.
For the maximum matching problem, such a sparsifier $H$ should preserve the size of the maximum matching of $G$ up to a factor of $1+\eps$.
It was shown in \cite{Sol18} that one can \emph{locally} compute a $(1+\eps)$-maximum matching sparsifier of degree $O(\alpha / \eps)$,
for any network of arboricity bounded by $\alpha$. All the sparsifiers of \cite{Sol18} adhere to a rather strict notion of locality,
which makes them applicable to several settings. In particular, for distributed networks, all the sparsifiers of \cite{Sol18} can be computed in a single round of communication.
The definition of a sparsifier for the minimum vertex problem is more involved, and we omit it here for conciseness (refer to \cite{Sol18} for the formal definition),
but the bottom-line is the same: For any distributed network of arboricity bounded by $\alpha$,
one can compute a $(1+\eps)$-minimum vertex cover sparsifier of degree $O(\alpha / \eps)$ in a single round.

Similarly to the maintenance of a maximal matching, maintaining these bounded degree sparsifiers dynamically do not require a sequential scan of the in-neighbors of a processor.
Indeed, these sparsifiers have s degree bound of $\Delta = O(\alpha / \eps)$ by definition, hence each processor can hold complete information on all
its adjacent edges that belong to the sparsifier, or equivalently, on all its corresponding neighbors. Following a deletion of an edge from the graph,
we first update the underlying representation.
 If the edge does not belong to the sparsifier, there is nothing special to do.
 Otherwise, we remove it from the sparsifier and check if another edge needs to be added to the sparsifier instead. In any case
we update the endpoints of the affected edges accordingly. It is straightforward to implement this update efficiently using the underlying representation.
Following an edge insertion, we may need to add it to the sparsifier, but this too involves a straightforward update.
In this way we can maintain bounded degree $(1+\eps)$-sparsifiers for maximum matching and minimum vertex cover using a local memory at processors that is (close to) linear in the network arboricity.

Subsequently, we can naively run static distributed algorithms for approximate maximum matching and minimum vertex cover on top of the bounded degree sparsifiers,
following every update step.
Due to the degree bound of the sparsifiers, in this way we adhere to the local memory constraints at processors.
To be able to run the distributed algorithm (following every update step), alas, we need to assume that all processors wake up prior to each such run,
which does not apply to the local wakeup model.  
Instead of running a static distributed algorithm from scratch on the sparsifiers following every update step,
we shall apply more efficient dynamic algorithms on top of the sparsifiers.

\cite{PS16} devised distributed algorithms for maintaining, in networks of degree bounded by $\Delta$,
$(1+\eps)$-approximate and $(3/2)$-approximate maximum matching with update time $O(1/\eps)$ and message complexities $(\Delta)^{O(1/\eps)}$ and $O(\Delta)$, respectively. (In fact, the bounds on the update time and message complexities hold in the worst-case.
Moreover, these algorithms extend to bounded arboricity graphs; refer to Corollary 3.1 in \cite{PS16}.)
Running these dynamic algorithms on top of the bounded degree $(1+\eps$)-maximum matching sparsifier that we maintain dynamically, we obtain the following result.
\begin{theorem} \label{thm1}
For any $\alpha \ge 1$, any arboricity $\alpha$ preserving sequence of edge and vertex updates starting from an empty graph and any $\eps > 0$,
there are distributed algorithms for maintaining $(1+\eps)$-approximate and $(3/2 + \eps)$-approximate maximum matching with   amortized update time
  $O(1/\eps + \log n)$ and   amortized message complexities of $(\alpha/ \eps)^{O(1/\eps)} + O(\log n)$ and $O(\alpha /\eps + \log n)$, respectively.
The local memory usage is $O(\alpha/\eps)$.
\end{theorem}

There is a straightforward distributed algorithm for maintaining a maximal matching, in networks of degree bounded by $\Delta$,
with update time $O(1)$ and message complexity $O(\Delta)$. Such an algorithm can be used to maintain a 2-approximate minimum vertex cover within the same bounds.
Running this dynamic algorithm
on top of the bounded degree $(1+\eps$)-minimum vertex cover sparsifier that we maintain dynamically, we obtain the following result.
\begin{theorem} \label{thm2}
For any $\alpha \ge 1$, any arboricity $\alpha$ preserving sequence of edge and vertex updates starting from an empty graph and any $\eps > 0$,
there is a distributed algorithm for maintaining a $(2+\eps$)-approximate minimum vertex cover with an amortized update time
of $O(\log n)$ and an amortized message complexity of $O(\alpha /\eps + \log n)$.  The local memory usage is $O(\alpha/\eps)$.
\end{theorem}
\section{The Flipping Game}
This section is devoted to the flipping game and its applications.

We start by proposing a generic paradigm for this game (Section \ref{appgen}).
In Section \ref{sec:equiv} we show a reduction from the edge orientation problem to the flipping game,
and in Section \ref{reduc} we show a reduction in the other direction, thus obtaining an equivalence.
Some applications of the flipping game are given in Section \ref{appl}.

\subsection{A Generic Paradigm for the Flipping Game} \label{appgen}
The flipping game provides a \emph{local} solution for the following  generic problem.
We want to maintain a dynamic graph $G$ in which each vertex has a value.
There are two types of updates to the graph: (1) edge insertion and deletion, (2) a change of a value at a vertex.
(We may also consider scenarios where there is only one type of updates. In particular, the scenario where the graph topology is static and vertex values are dynamic
is already not trivial.)
A query specifies a vertex $v$ and to answer it we need to compute some fixed function of the values of $v$ and its neighbors.

We restrict ourselves to a natural family $\mathcal{F}$ of algorithms that maintain an edge orientation of $G$, where each vertex $v$ maintains the current values of all its \emph{in-neighbors (incoming neighbors)}.
When the value of a vertex $v$ changes, $v$ transmits its new value to all its out-neighbors.
When a vertex $v$ is queried, $v$ collects the values of its out-neighbors, computes the function and returns the result.
The algorithm has the freedom to change the edge orientation by flipping edges.
The cost of flipping an edge outgoing of $v$ is $0$ if we flip it during a query or update at $v$, and $1$ otherwise.
(Note that the algorithms of \cite{NS13,KKPS14,HTZ14} can also be viewed as belonging to $\mathcal{F}$, but
they all require that the outdegree of all vertices at all times will be bounded by some threshold $\Delta$.
In general, the algorithms of $\mathcal{F}$ may violate this requirement.)

The \emph{(communication) cost} of an algorithm $A$ in this family for serving a sequence of operations
$\sigma$ is $$c(A,\sigma) = t+f+\sum_{op\in \sigma ~\mid~ \mbox{\small op updates or queries v}} outdegree(v) ,$$ where $t$ is the number of edge insertions and deletions in $\sigma$, $f$ is the cost of edge flips that the algorithm performs during $\sigma$, and the sum is over all vertex updates and queries in $\sigma$
of the outdegree of the vertex $v$ to which the operation ($op$) applies.
We remark that this cost $c(A,\sigma)$ is equal to the total runtime of algorithm $A$ with respect to $\sigma$, up to a constant factor.
(To be accurate, the runtime should include the cost of extracting the relevant information on the incoming neighbors of the queried vertices.
If this cost is high, which depends on the application,   that application cannot be solved using our scheme.)

The \emph{flipping game} is a particular algorithm in $\mathcal{F}$ that resets a \emph{vertex} $v$ whenever we apply a query or update to $v$,
which means that all the outgoing edges of $v$ are flipped and become incoming to $v$.
The flipping game is simple and local. Furthermore, it is easy to verify that for any sequence of operations,
the cost of the flipping game is at most twice the cost of any other algorithm in $\mathcal{F}$. Hence:
\begin{observation}
Denote the flipping game algorithm by $R$.
For any sequence of operations $\sigma$ and   algorithm $A \in \mathcal{F}$, $c(R,\sigma)\le 2 \cdot c(A,\sigma)$.
The initial graph may be arbitrary (non-empty), but $R$ and $A$ should start from the same edge orientation.
\end{observation}
\begin{proof}
Since $R$ always flips edges at $0$ cost, the total cost of $R$ is
$$c(R,\sigma) = t+\sum_{op\in \sigma ~\mid~ \mbox{\small op updates or queries v}} outdegree(v) \ .$$

Consider an edge $e=(u,v)$ and
an operation at $v$ during which $e$ was outgoing of $v$ (and therefore $R$ was charged for the communication along $e$).
If this is the first operation
in which $R$ is charged for $e$ then either $A$ is charged for $e$ during this operation as well, or $A$ flipped $e$ before this operation.
If there was a previous operation in which $R$ was charged for $e$ then it must have been an operation at $u$.
So it must be the case that either $A$ flipped the edge between the operation at $u$ and the operation at $v$ or $A$ paid for $e$ in at least one of these operations.
\QED
\end{proof}

\subsection{A reduction from the edge orientation problem to the flipping game} \label{sec:equiv}
We can easily simulate the BF algorithm using the reset operations of the flipping game.
The following lemma shows that for an appropriate outdegree threshold the
amortized time per edge update of the simulation is essentially the same as the amortized time per operation (update or reset)
of the flipping game. Thus the amortized bound of the flipping game is essentially as large as that of the BF algorithm.

\begin{lemma} \label{cl:reduction}
Consider an arbitrary sequence of $t$ edge updates,
and suppose that the flipping game (either the basic game or the $\Delta$-flipping game)
on this update sequence with any $r$ resets performs at most $k(t + r)$ edge flips, for any parameter $r$.
Then for any $\Delta \ge k$, the BF algorithm
with outdegree threshold $\Delta$, performs at most $(k t)/(1 - k/(\Delta+1))$ edge flips.
\end{lemma}
\begin{proof}
We simulate the BF algorithm using the flipping game by resetting every vertex  whose outgoing edges are flipped by the reset cascade of the BF
algorithm. 
Let $r$ be the total number of resets that the simulation performs and let $f$ be the total number of edge flips.
Since each reset of the simulation flips at least $\Delta+1$ edges, $r\le f/(\Delta+1)$.
By our assumption on the flipping game we have  $f\le k(t + r)$.
The lemma follows by substituting the upper bound on $r$ into this inequality and rearranging.
\QED
\end{proof}
For example, if we set $\Delta = 2k-1$, the amortized update time of the simulation (per edge update), and hence of the BF algorithm,
is at most $2k$. This shows that we only lose a factor of 2 when amortizing over the edge updates rather than over both the edge updates and the reset operations.

\subsection{A reduction from the flipping game to the   edge orientation problem} \label{reduc}
\begin{lemma}  \label{basicflip}
Suppose we can maintain a $\Delta$-orientation for some sequence of $t$ edge updates while doing $f$ edge flips, starting with the empty graph.
Then the flipping game on this update sequence with any $r$ resets performs at most $t + f + 2\Delta r$ edge flips,
for any $r$.
\end{lemma}
\begin{proof}
We charge the edge flips performed by reset operations of the flipping game to edge flips performed to maintain the $\Delta$-orientation.
Following  a reset on $v$, we place two tokens on every  edge that is outgoing of $v$ in the
$\Delta$-orientation. When the $\Delta$-orientation flips an edge $e$  we place a  token on $e$.
When an edge $e$ is inserted to the graph we place a token on $e$.
The total number of tokens placed on edges is $t + f + 2\Delta r$.
We claim that the number of tokens
 placed on $e$ is no smaller than the number of times $e= (u,v)$ flips in the flipping game (so these tokens ``pay'' for these flips).
Consider
a maximal sequence $\sigma$ of flips of $e$ that occur while the orientation of $e$ in
the $\Delta$-orientation does not change. Assume without loss of generality that $e$ is oriented from
$u$ to $v$ by  the $\Delta$-orientation during  $\sigma$.
Let $x$ be the number of flips in  $\sigma$.
During the time span of  $\sigma$ both $u$ and $v$ were reset at least $\floor{x/2}$ times.
Each such reset of $u$ places $2$ tokens on $e$. The total number of these tokens is at least $x-1$.
The flip of $e$ performed by the $\Delta$-orientation or its insertion just before $\sigma$ starts contributes an additional token.
\QED
\end{proof}

The number of edge flips per edge update performed for maintaining the $\Delta$-orientation is $f/t$
whereas the number of edge flips per operation of the flipping game is
$(t + 2f + 2\Delta r)/(t + r) = O(\max\{\Delta, f/t\})$.
Thus the flipping game does not depend on $\Delta$ but its amortized time bound does depend on $\Delta$.

To remove the dependency of the amortized time of the flipping game on
 the outdegree threshold $\Delta$, we modify the game slightly and make  it aware of $\Delta$ as follows.
We define the \emph{$\Delta$-flipping game} in which when we reset a vertex we flip all its outgoing edges only if
there are more than  $\Delta$ such edges.
Note that by setting $\Delta' = 3\Delta-1$,
we get that the total number of flips of the $\Delta'$-flipping game is at most $3(t+f)$.
This bound is the same, up to a constant, as for maintaining the $\Delta$-orientation, even though we also performed $r$ reset operations.
%
\begin{lemma} \label{boundedflip}
Suppose we can maintain a $\Delta$-orientation for some sequence of $t$ edge updates while doing $f$ edge flips, starting with the empty graph.
Then the $\Delta'$-flipping game on this update sequence with any $r$ resets performs at most $(t+f) (\Delta' + 1)/ (\Delta' + 1 - 2\Delta)$ edge flips,
for any parameters $r$ and $\Delta' \ge 2\Delta$.
\end{lemma}
\begin{proof}
Our proof uses a potential function argument similar to the one used in Lemma 1 of \cite{BF99}.
We define an edge to be  \emph{good} if its orientation in the flipping game is the same as in
the $\Delta$-orientation and \emph{bad} otherwise.
We define the potential $\Psi$ to be
 the number of bad edges in the current graph.
Initially $\Psi = 0$.
 Each
insertion or a flip performed by the $\Delta$-orientation increases $\Psi$ by at most one, while edge deletions may only decrease $\Psi$.

Consider a reset of some vertex $v$ of outdegree greater than $\Delta'$.
By the definition of a $\Delta$-orientation at most $\Delta$ of $v$'s outgoing edges are good.
As a result of the flip these $\Delta$ edges may become bad, but
at least $\Delta' + 1 - \Delta$ edges were bad and become good.
It follows that as a result of the reset $\Psi$ decreases by at least $\Delta' + 1 - 2\Delta$.
This implies that
the total number of reset operations on vertices  with outdegree greater than $\Delta'$ is at most $(t + f) / (\Delta' + 1 - 2\Delta)$.
The total number of times a good edge becomes bad due to the resets is bounded by $\Delta (t + f) / (\Delta' + 1 - 2\Delta)$,
from which we conclude that the total number of times a bad edge becomes good due to the resets is bounded by $t + f + \Delta (t + f) / (\Delta' + 1 - 2\Delta)$.
Summarizing, the total number of flips made by the flipping game is bounded by
$$(t + f) (1 + 2\Delta/ (\Delta' + 1 - 2\Delta)) ~=~ (t+f) (\Delta' + 1)/ (\Delta' + 1 - 2\Delta).$$
\end{proof}
%
\subsection{Applications} \label{appl}
As discussed in the introduction,
by using the flipping game instead of the BF algorithm, we obtain \emph{local} algorithms for several dynamic graph problems.
In this section we describe two such applications to the problems of
dynamic maximal matching and adjacency queries.
\vspace{2pt}
\\
{\bf Dynamic maximal matching.~}
The goal here is to maintain a maximal matching $M$ in a graph $G$ that undergoes edge insertions and deletions.
Following an edge insertion or a deletion of an edge  not in $M$, a maximal matching remains maximal.
The difficult operation is a deletion of an edge in $M$.
Following an edge deletion $(u,v)\in M$ both $u$ and $v$ become free, and if either $u$ or $v$ has a free neighbor then $M\setminus \{(u,v)\}$ is
not maximal anymore, and we must add edges from $u$ and $v$ to one of their free neighbors.

Neiman and Solomon \cite{NS13}
reduced this problem to the edge orientation problem as follows.
We maintain an edge orientation of $G$, and each vertex $v$ maintains its free incoming neighbors.
Following an edge deletion $(u,v)\in M$,   $u$ and $v$ perform the following operations. (We restrict attention to $u$ and describe what it does; $v$ performs the same operations.)
First $u$ notifies its out-neighbors that it is free. Then it checks whether its list of free in-neighbors is not empty.
If $u$ has a free in-neighbor $x$ then we add  the edge $(u,x)$ to $M$ and both $x$ and $u$ notify their out-neighbors that
they are now matched. Otherwise $u$ scans its out-neighbors for a free vertex. If $u$ finds a free out-neighbor $x$ then we add
$(x,u)$ to $M$ and both $x$ and $u$ notify their out-neighbors that
they are matched.

This reduction implies that from an algorithm that maintains a $\Delta$-orientation with an
update time of $T$ (either amortized or worst-case), we can get a dynamic algorithm for
maximal matching with an update time of $O(\Delta + T)$ (again, either amortized or worst-case).

The result of \cite{HTZ14}
shows that in a graph with arboricity bounded by $\alpha$ the BF algorithm
maintains an $O(\beta \alpha)$-orientation in amortized update time of
$T = O(\log(n/(\beta \alpha))/ \beta)$ for any parameter $\beta \ge 1$. (Refer to App.\ \ref{discuss} for more details.)
Using this tradeoff in the particular case of $\Delta \approx T$ (where $\Delta = \Theta(\beta \alpha)$),
we get a dynamic algorithm for maximal matching with $O(\alpha + \sqrt{\alpha \log n})$ amortized update time.
The drawback of the resulting algorithm is that it is not local. Indeed, this is because any algorithm for maintaining
$\Delta$-orientation is inherently non-local.

To get a local algorithm for dynamic maximal matching we use our (inherently local) flipping game.
As before, we maintain  an orientation and each vertex maintains its free in-neighbors.
But now, when a vertex $v$ scans its out-neighbors (either when $v$ changes its state from matched to unmatched or vice versa, or when $v$ looks for a free out-neighbor),
then we also \emph{reset} $v$, thereby flipping all its outgoing edges.
The  total running time of the resulting local algorithm for dynamic maximum matching
 is linear in the number of edge flips made by the underlying flipping game.

To bound the number of edge flips made by the flipping game,
note that we reset at most a constant number of vertices per edge update.
By Lemma \ref{basicflip}, combined with the result of \cite{HTZ14} for the case $\Delta \approx T$,
we conclude that the amortized number of flips made by the flipping game is  $O(\alpha + \sqrt{\alpha \log n})$.

The flipping game can be easily distributed. Resetting a vertex requires one communication round,
and the message complexity is asymptotically the same as the runtime in the centralized setting.
Summarizing, we have proved the following result. 
\begin{theorem} \label{localmaxmatch}
For any arboricity $\alpha$ preserving sequence, there is a \emph{local} algorithm for maintaining a maximal matching
on the corresponding dynamic $n$-vertex graph $G$ with an amortized update time of $O(\alpha + \sqrt{\alpha \log n})$.
The space usage of the algorithm is linear in the graph size.
Moreover, there is a distributed algorithm for maintaining a maximal matching with
an amortized message complexity of $O(\alpha + \sqrt{\alpha \log n})$ and a constant worst-case update time.
\end{theorem}

\noindent
\vspace{2pt}
{\bf Adjacency queries.}
In this application we want to maintain a
deterministic linear space data structure that allows efficient adjacency queries in a dynamic graph.
(If we use dynamic perfect hash tables to represent adjacency lists then the data structure is of linear size but randomized.)
Although the problem of supporting adjacency queries is inherently local, the state-of-the-art deterministic solution 
(described next) relies on the inherently non-local task of maintaining a low outdegree orientation.


The BF algorithm with outdegree threshold $O(\alpha)$ has an amortized update time of $O(\log n)$.
Such an orientation allows to support adjacency queries in $O(\alpha)$ worst-case time, since to decide if
the graph contains the edge $(u,v)$, it suffices to search $u$ among the out-neighbors of $v$, and $v$ among the out-neighbors of $u$.
Later Kowalik~\cite{Kowalik07} proved that for outdegree threshold $O(\alpha  \log n)$, the amortized update time of the BF algorithm is constant.
Kowalik noted that if the out-neighbors of each vertex are stored
in a  balanced search tree, then the amortized update time increases from $O(1)$ to $O(\log \alpha + \log\log n)$ (each edge flip requires an insertion to
and a deletion from a balanced search tree, and similarly for edge insertions) but the worst-case query time becomes $O(\log \alpha + \log\log n)$.
When the arboricity bound is polylogarithmic in $n$, these bounds are $O(\log \log n)$, and
using more sophisticated data structures, one can improve this bound to $O(\log\log \log n)$ under the RAM model.  

Next, we describe a local data structure for supporting adjacency queries.  
To this end we use the $\Delta$-flipping game, for $\Delta = O(\alpha  \log n)$.
Specifically, to perform an adjacency query with $(u,v)$, we start by
 resetting  $u$ and $v$, thereby flipping the out-neighbors of $u$ (resp.\ $v$) if it has more than $\Delta$ out-neighbors.
Following these resets, $u$ and $v$ have at most $\Delta$ out-neighbors
and we answer the  query by scanning these lists of out-neighbors as before.
To speed up the query further
we keep the out-neighbors of each vertex $v$ with at most $\Delta$ out-neighbors in a balanced search tree as described above.
(More concretely, we start building the tree at $v$ when $v$'s outdegree drops below $2\Delta$ and once we have the tree ready we maintain it as
long as the outdegree of $v$ is below $2\Delta$. This guarantees that we always have a tree ready when the outdegree is at most $\Delta$,
while keeping the cost of constructing the trees in check.)  

By Lemma \ref{boundedflip} combined with the result of \cite{Kowalik07},
 the amortized number of edge flips made by the $\Delta$-flipping game is constant.
Hence both adjacency queries and edge updates take $O(\log \alpha + \log\log n)$ amortized time.
So our $\Delta$-flipping game provides a local data structure for adjacency queries at the cost of
having only an amortized guarantee for the query time rather than a worst-case guarantee.
Summarizing, we have proved the following result.
\begin{theorem}  \label{localadj}
For any arboricity $\alpha$ preserving sequence, there is a (deterministic) \emph{local} algorithm for  supporting adjacency queries
in the corresponding dynamic $n$-vertex graph $G$ with an amortized update time of $O(\log \alpha + \log\log n)$.
The space usage of the algorithm is linear in the graph size.
\end{theorem}





\vspace{16pt}
\pagenumbering{roman}
\appendix
\centerline{\LARGE\bf Appendix}


\section{More on the Edge Orientation Problem in Centralized Networks} \label{discuss}
In light of the asymptotic optimality of the BF algorithm discussed in Section \ref{sec131}, any existential bound for the problem translates into an algorithmic result with the same asymptotic guarantees on the outdegree and the amortized update time.
\cite{Kowalik07} proved an existential bound of $O(\alpha \log n)$-orientation with $O(1)$ amortized update time.
\cite{HTZ14} proved a general existential tradeoff:  $O(\beta \alpha)$-orientation with $O(\log(n/(\beta \alpha))/ \beta)$ amortized update time, for any   $\beta \ge 1$;
note that the results of \cite{BF99} and \cite{Kowalik07} provide the two extreme points on the tradeoff curve of \cite{HTZ14}.
The tradeoff of \cite{HTZ14} in the particular case of $\alpha = O(1)$ and $\beta = \sqrt{\log n}$ bounds both the outdegree and the amortized update time by $O(\sqrt{\log n})$; nevertheless, to maintain constant outdegree (when $\alpha = O(1)$), the state-of-the-art update time is still $O(\log n)$, due to BF.

The edge orientation problem with worst-case time bounds was first studied in \cite{KKPS14}, where it was shown that one can maintain a $\Delta$-orientation with $O(\beta \alpha \Delta)$ \emph{worst-case} update time, for $\Delta = \inf_{\beta > 1} \{\beta \alpha + \lceil \log_\beta n \rceil\}$.
(A similar result was obtained by \cite{HTZ14}.)
\cite{BB17} presented a tradeoff of $O(\frac{\alpha \cdot \log^2 n}{\beta})$-orientation with $O(\beta)$ worst-case update time, for any $\beta = O(\log n)$, along with additional refinements over the previous work \cite{KKPS14,HTZ14}. We remark that the worst-case guarantees of \cite{KKPS14,HTZ14,BB17} are inferior to the aforementioned amortized guarantees, and in the particular case of $\alpha= O(1)$, none of these results provides an outdegree lower than $O(\frac{\log n}{\log\log n})$, even for a polynomial worst-case update time.

\subsection{Some applications of the edge orientation problem}
In this section we provide a very short (and non-exhaustive) overview on some of the applications of the edge orientation problem in the context of dynamic graph (centralized) algorithms. For a more detailed account on these applications, we refer to \cite{NS13,KKPS14,HTZ14,PS16}.

\cite{NS13} showed a reduction from the problem of maintaining maximal matching to the edge orientation problem.
Specifically, if a $\Delta$-orientation can be maintained within update time $T$ (either amortized or worst-case),
then a maximal matching can be maintained within update time $O(T + \Delta)$ (again, either amortized or worst-case).
\cite{NS13} plugged the tradeoff of BF into this reduction, and obtained an amortized update time of $O(\log n / \log\log n)$ for maintaining maximal matching
in graphs of  low arboricity.
By plugging their own improved tradeoff, \cite{HTZ14} reduced the amortized update time to $O(\sqrt{\log n})$.
A worst-case update time of $O(\log n)$ for this probelm was obtained by \cite{KKPS14}, using their result for the edge orientation problem.
Note also that a maximal matching naturally translates into a 2-approximate vertex cover, and this translation can be easily maintained dynamically.
The edge orientation problem of \cite{BF99} was shown to be useful also in other dynamic graph problems,
such as distance oracles, approximate matching, and coordinate queries; see \cite{KK03,KKPS14,HTZ14,BS15,BS16} for more  details.





\end{document}

%% file: latexmac.tex
\def\denseformat{
\setlength{\textheight}{9.5in}
\setlength{\textwidth}{6.9in}
\setlength{\evensidemargin}{-0.3in}
\setlength{\oddsidemargin}{-0.3in}
\setlength{\headsep}{10pt}
\setlength{\topmargin}{-0.44in}
\setlength{\columnsep}{0.375in}
\setlength{\itemsep}{0pt}
}

\newtheorem{theorem}{Theorem}[section]

\newtheorem{claim}[theorem]{Claim}
\newtheorem{lemma}[theorem]{Lemma}

\newtheorem{corollary}[theorem]{Corollary}

\newtheorem{observation}[theorem]{Observation}

\def\boldhead#1:{\par\vskip 7pt\noindent{\bf #1:}\hskip 10pt}
\def\ithead#1:{\par\vskip 7pt\noindent{\it #1:}\hskip 10pt}

\def\floor#1{\lfloor #1\rfloor}
\def\inline#1:{\par\vskip 7pt\noindent{\bf #1:}\hskip 10pt}
\def\midinline#1:{\par\noindent{\bf #1:}\hskip 10pt}
\def\dnsinline#1:{\par\vskip -7pt\noindent{\bf #1:}\hskip 10pt}
\def\ddnsinline#1:{\newline{\bf #1:}\hskip 10pt}
\def\largeinline#1:{\par\vskip 7pt\noindent{\large\bf #1:}\hskip 10pt}
\long\def\commhide #1\commhideend{}
\long\def\commtim #1\commendt{#1}
\long\def\commb #1\commbend{}
\long\def\commedit #1\commeditend{} 

\long\def\commB #1\commBend{}       

\long\def\commex #1\commexend{}     

\long\def\commsiena #1\commsienaend{}  

\long\def\commBI #1\commBIend{}  

\long\def\CProof #1\CQED{}

\def\blackslug{\hbox{\hskip 1pt \vrule width 4pt height 8pt
    depth 1.5pt \hskip 1pt}}
\def\QED{\quad\blackslug\lower 8.5pt\null\par}

\long\def\PPP#1{\noindent{\bf Proof:}{ #1}{\quad\blackslug\lower 8.5pt\null}}

\long\def\denspar #1\densend
{#1}

\newif\ifnotesw\noteswtrue
   {\ifnotesw\marginpar[\hfill\(\top\)]{\(\top\)}\fi}%
   {\ifnotesw\marginpar[\hfill\(\bot\)]{\(\bot\)}\fi}

\newcommand{\mnote}[1]%
    {\ifnotesw\marginpar%
        [{\scriptsize\it\begin{minipage}[t]{\marginparwidth}
        \raggedleft#1%
                        \end{minipage}}]%
        {\scriptsize\it\begin{minipage}[t]{\marginparwidth}
        \raggedright#1%
                        \end{minipage}}%
    \fi}

\def\MathF{\hbox{\rm I\kern-2pt F}}
\def\MathP{\hbox{\rm I\kern-2pt P}}
\def\MathR{\hbox{\rm I\kern-2pt R}}
\def\MathZ{\hbox{\sf Z\kern-4pt Z}}
\def\MathN{\hbox{\rm I\kern-2pt I\kern-3.1pt N}}
\def\MathC{\hbox{\rm \kern0.7pt\raise0.8pt\hbox{\footnotesize I}
\kern-4.2pt C}}
\def\MathQ{\hbox{\rm I\kern-6pt Q}}

\newsavebox{\ttop}\newsavebox{\bbot}

\def\eps{\epsilon}

